\documentclass[fleqn]{LMCS}

\def\eqalign#1{\null\,\vcenter{\openup\jot\mathsurround=0 pt
  \ialign{\strut\hfil$\displaystyle{##}$&$\displaystyle{{}##}$\hfil
      \crcr#1\crcr}}\,}
\def\iff\ {if\null f\ }

\usepackage{hyperref,amsmath,amsfonts,amssymb,enumerate}  
\usepackage{amsbsy}
\usepackage{proof,stmaryrd,wrapfig,version}
\usepackage[all]{xy} \CompileMatrices

\newcommand{\app}{\mathtt{app}}
\renewcommand{\delta}{}

\newcommand{\la}{\pmb{\lambda}}

\newcommand{\CC}{\mathbb{C}}
\newcommand{\DD}{\mathbb{D}}
\newcommand{\K}{\mathbf K}
\newcommand{\SSS}{\mathbf S}
\newcommand{\A}{\mathbf A}

 \newcommand{\reguno}[2]
  {
    \frac{\textstyle #1}{\textstyle #2} 
  }
\newcommand{\utrans}[1]{\mathrel{\smash{\stackrel{#1}{\longrightarrow}}}}

\newcommand{\ctrans}[1]{\mathrel{\smash{\stackrel{#1}{\longrightarrow_C}}}}
\newcommand{\csim}{\mathrel{\sim_C}}
\newcommand{\itrans}[1]{\mathrel{\smash{\stackrel{#1}{\longrightarrow_I}}}}
\newcommand{\isim}{\mathrel{\sim_I}}

\newcommand{\wutrans}[1]{\mathrel{\smash{\stackrel{#1}{\Longrightarrow}}}}

\newcommand{\witrans}[1]{\mathrel{\smash{\stackrel{#1}{\Longrightarrow_I}}}}
\newcommand{\wisim}{\mathrel{\approx_{I}}}

\newcommand{\wrsim}{\mathrel{\approx_{R}}}
\newcommand{\wsim}{\mathrel{\approx}}

\newcommand{\eqlcl}{\mathrel{\simeq_{l}}}
\newcommand{\dlcl}{\mathrel{\downarrow_{l}}}

\newcommand{\eqll}{\mathrel{\approx_{l}}}
\newcommand{\dll}{\mathrel{\Downarrow_{l}}}

\def\doi{5 (3:6) 2009}
\lmcsheading%
{\doi}
{1--35}
{}
{}
{Aug.~13, 2008}
{Aug.~\phantom{0}6, 2009}
{} 

\begin{document}
\title{RPO, Second-order Contexts, and
  \texorpdfstring{$\lambda$}{Lambda}-calculus\rsuper*}

\titlecomment{{\lsuper*}Work supported by
ART PRIN Project prot.\ 2005015824 and by FIRB Project RBIN04M8S8 (both funded by MIUR)}

\author[P.~Di Gianantonio]{Pietro Di Gianantonio}
\author[F.~Honsell]{Furio Honsell}
\author[M.~Lenisa]{Marina Lenisa}
 \address{Dip. di Matematica e Informatica,
  Universit\`a di Udine\\
  via delle Scienze 206, 33100 Udine, Italy}
  \email{\{digianantonio,honsell,lenisa\}@dimi.uniud.it}


\begin{abstract}
First, we extend  Leifer-Milner RPO theory, by giving
general conditions to obtain IPO labeled transition systems (and bisimilarities)
with a reduced set of transitions, and possibly finitely branching. Moreover, 
we study the weak variant of Leifer-Milner theory, by giving general conditions under
which  the weak bisimilarity is a congruence. Then, we apply  such extended 
RPO technique to the lambda-calculus, endowed with lazy and call by value 
reduction strategies. 
 We show that,  contrary to  process calculi, one can deal directly  with 
the lambda-calculus syntax and apply  Leifer-Milner technique to a category 
of contexts, provided that we work in the framework of weak bisimilarities.
 However,  even in the case of the transition system with minimal contexts, 
the resulting bisimilarity is 
infinitely branching, due to the fact that, in standard context categories, 
parametric rules such as the beta-rule can be represented only by infinitely
many ground rules.
 To overcome this problem, we introduce the general notion of second-order
context category. We show that, by carrying out the RPO construction in this 
setting, the lazy observational equivalence can be captured 
as a weak bisimilarity equivalence on a finitely branching
transition system. This result is achieved  by considering an encoding of
lambda-calculus in  Combinatory Logic.
\end{abstract}

\maketitle
\baselineskip=14 pt
\section{Introduction}
Recently, much attention has been devoted to derive \emph{labeled transition systems}
 and \emph{bisimilarity congruences} from \emph{reactive systems}, 
 in the context of process
 languages and graph rewriting, \cite{Sew02,LM00,SS03,GM05,BGK06,BKM06,EK06}. 
In the theory of process algebras, the operational semantics of CCS was originally given via a labeled
transition system (lts), while  more recent
process calculi have been presented via  reactive systems plus structural rules.
Reactive systems naturally induce behavioral equivalences which are congruences w.r.t. contexts,
while lts's naturally induce bisimilarity equivalences with  coinductive characterizations.
However, such equivalences are not congruences in general, or else it is an heavy, ad-hoc  
task to prove that they are congruences.

Generalizing \cite{Sew02}, Leifer and Milner \cite{LM00} presented a general categorical 
method for deriving 
a transition system from a reactive system, in such a way that the induced bisimilarity is a congruence.
The labels in Leifer-Milner's  transition system are those contexts which are  \emph{minimal} for a given reaction to fire. 
Minimal contexts are identified via the categorical notion of \emph{relative pushout (RPO)}.
Leifer-Milner's central result guarantees that, under a suitable categorical condition,
 the induced bisimilarity is a \emph{congruence} w.r.t. all contexts.  

In the literature, some case studies have been carried out, especially in the setting of process calculi, for testing the expressivity of 
Leifer-Milner's approach. Some difficulties have arisen in applying the approach directly to such
languages, viewed as   
Lawvere theories, because of  structural rules.
To overcome this problem, two different approaches have been considered. The first
approach consists in using 
 more complex categorical constructions, where structural rules are accounted for explicitly, 
  \cite{Lei01,SS03,SS05}. In the second approach,  intermediate encodings have been considered
  in 
 graph theory, for which the approach of  ``borrowed contexts''  has been developed  \cite{EK06},
 and in Milner's bigraph theory. Here structural rules are avoided, since structurally equivalent terms
 are equated in the target  language.
 
 Moreover, the following further issues have arisen in applying
 Leifer-Milner's technique.  \begin{enumerate}[(i)] \item
 Leifer-Milner's bisimilarity is still redundant, and many labels have
 to be eliminated \emph{a posteriori}, by an \emph{ad-hoc}
 reasoning. Thus general results are called for, in order to reduce
 the complexity of the bisimilarity \emph{a priori}.  \item In some
 cases it is useful to consider \emph{weak} variants of Leifer-Milner
 technique. However, for the \emph{weak bisimilarity} we only have a
 \emph{partial congruence} result, stating that such bisimilarity is a
 congruence w.r.t. a certain class of contexts. However, in many
 concrete cases, the weak bisimilarity turn out to be a \emph{full
 congruence}. Thus it will be useful to study general conditions under
 which this happens.  \item When Leifer-Milner technique is applied in
 the standard setting of term and context categories (Lawvere
 theories), the rules in the rewriting system cannot be represented
 parametrically, but only at a \emph{ground level} through a
 (infinite) series of possible instantiations. As a consequence, the
 bisimilarity turns out to be infinitely branching.  In \cite{KSS05},
 a generalization of Leifer-Milner technique for dealing with
 parametric rules has been introduced. This approach is rather complex
 and not completely satisfactory. An alternative approach (which is
 considered in the present paper) consists in studying
 \emph{second-order} versions of term and context categories, which
 allow \emph{parametric} representations of rewriting rules, and
 carrying out Leifer-Milner technique in this setting.
 \end{enumerate}
 
\noindent In this paper, we address all the above issues. In
particular, in the first part of the paper, we extend Leifer-Milner
theory, by providing general results for reducing the complexity of
the bisimilarity, and by studying conditions under which the weak
bisimilarity is a full congruence.  Then, we focus on the prototypical
example of reactive system given by the $\lambda$-calculus, endowed
with lazy and call by value (cbv) reduction strategies. We show that,
in principle, contrary to most of the case studies considered in the
literature, one could deal directly with the $\lambda$-calculus syntax
and apply Leifer-Milner technique to the category of term contexts
induced by the $\lambda$-terms, provided that we work in the setting
of \emph{weak bisimilarities}.  Applying our general results, we get
quite economical weak bisimilarities which are congruences and we
recover exactly both lazy and cbv contextual equivalences. As a
by-product, we also get an alternative proof of the Context Lemma for
the lazy case.  However, the bisimilarities that we obtain are still
infinitely branching.  This is mainly due to the fact that, in the
category of contexts, the $\beta$-rule cannot be described
parametrically, but it needs to be described extensionally using an
infinite set of pairs of ground terms.
 In order to overcome this problem, we consider the combinatory logic
 and we introduce the general notion of \emph{category of second-order
   term contexts}, which provide a solution to the third issue above.
 Our main result amounts to the fact that, by carrying out
 Leifer-Milner's construction in this setting, the \emph{lazy
   contextual equivalence} can be captured as a \emph{weak
   bisimilarity equivalence} on a (finitely branching) transition
 system, while for the cbv case, the finitely branching transition
 system induces a bisimilarity which is strictly included in the
 contextual equivalence.  Technically, these results are achieved by
 considering an encoding of the lazy (cbv) $\lambda$-calculus in KS
 Combinatory Logic (CL), endowed with a lazy (cbv) reduction strategy,
 and by showing that the lazy (cbv) contextual equivalence on
 $\lambda$-calculus can be recovered as a lazy (cbv) equivalence on
 CL.  It is necessary to consider such encoding, since the approach of
 second-order context categories proposed in this paper works for
 reaction rules which are ``local'', that is, the reaction does not act
 on the whole term, but only locally.  But the substitution operation
 on $\lambda$-calculus is not local.

 Finally,  the correspondence results obtained in this paper about the observational equivalences on
$\lambda$-calculus and CL are interesting \emph{per se} and, although natural and ultimately
elementary, had not appeared previously in the literature.

\subsubsection*{Summary.} In Section~\ref{trs}, we summarize the theory of reactive systems
of \cite{LM00}. 
In Section~\ref{ext}, we extend such theory with new general results about weak bisimilarity, and
about
the ``pruning'' of Leifer-Milner lts and the induced bisimilarity.
In Section~\ref{lamc}, we present the $\lambda$-calculus together with lazy
and cbv reduction strategies and observational equivalences, and we discuss the RPO
approach applied to the $\lambda$-calculus endowed with a structure of context category. In 
Section~\ref{CL},  we focus on Combinatory Logic (CL), we show how to recover 
on CL  the lazy and cbv strategies and observational equivalences, and we discuss the RPO
approach applied to CL, viewed as a context category. In Section~\ref{sec}, we introduce the notion
of second-order context category, and we apply the RPO approach to CL  viewed as a second-order
rewriting system, thus obtaining a characterization of the  lazy  observational 
equivalence as a weak bisimilarity on a finitely branching lts. Final remarks and directions for 
future work appear in Section~\ref{final}. 
\smallskip

The present paper extends \cite{DHL08}. The main new contribution of the present paper is
the extension of  
Leifer-Milner theory, which appears in Section~\ref{ext}. This allows to deal with the 
$\lambda$-calculus in the subsequent sections in a smoother way, to get stronger results about
the lts and the induced bisimilarity, both for the lazy and for the cbv case, and also to provide an
alternative proof of  the 
Context Lemma in  the lazy case. 

\subsubsection*{Acknowledgments.} The authors thank the referees for many useful comments, 
which helped in greatly improving the paper.

\section{The Theory of Reactive Systems}\label{trs}
In this section, we summarize the theory of reactive systems proposed in \cite{LM00} to derive
lts's and bisimulation congruences from a given reduction semantics.
Moreover, we discuss weak variants of Leifer-Milner's bisimilarity equivalence.

The theory of \cite{LM00} is based on a categorical formulation of the notion of \emph{reactive system},
whereby \emph{contexts} are modeled as arrows of a category, \emph{terms} are  arrows having as
domain $0$ (a special object which denotes no holes), and reaction rules are pairs of terms.

\begin{defi}[Reactive System] \label{rslm}
A \emph{reactive system} ${\mathbf C}$ consists of:
\begin{enumerate}[$\bullet$]
\item a  category $\mathcal{C}$;
\item a distinguished object $0\in |\mathcal{C}|$;
\item a composition-reflecting subcategory $\mathcal{D}$ of \emph{reactive contexts};
\item a set of pairs ${\mathbf R} \subseteq \bigcup_{I \in |\mathcal{C}| } \mathcal{C}[0,I] \times \mathcal{C}[0,I] $ of
\emph{reaction rules}.
\end{enumerate}
\end{defi}

\noindent 
The reactive contexts are those in which a reaction can occur. By
composition-reflecting we mean that $dd'\in \mathcal{D}$ implies $d,d'
\in \mathcal{D}$.

Reactive systems on term languages can be viewed as a special case of
reactive systems in the sense of Leifer-Milner by instantiating
$\mathcal{C}$ as a suitable category of term and contexts, also called
the (free) Lawvere category, \cite{LM00}.  In this view, we often call
terms the arrows with domains $0$, and contexts the other arrows.

From the set of reaction rules one generates the reaction relation by
closing them under all reactive contexts:

\begin{defi}[Reaction Relation]
Given a reaction system with reactive contexts $\mathcal{D}$ and reaction rules ${\mathbf R}$, the
\emph{reaction relation} $\rightarrow$ is defined by:
\[t \rightarrow u 
  \quad\hbox{iff}\quad
  t=dl,\  u=dr
  \quad\hbox{for some}\quad
  d\in \mathcal{D} 
  \quad\hbox{and}\quad
  \langle l,r \rangle\in {\mathbf R}\ .
\]
\end{defi}
The behavior of a reactive system is expressed as an unlabeled transition system. On the other 
hand, many useful behavioral equivalences are only defined for lts's. 
The passage from reactive systems to lts's is obtained as follows.

\begin{defi}[Context Labeled Transition System]
  Given a reactive system ${\mathbf C}$, the associated context lts is defined as follows:
\begin{enumerate}[$\bullet$]
\item  states: arrows $t: 0\rightarrow I$ in $\mathcal{C}$, for any $I$;
\item  transitions: $t \ctrans{c} u $ \iff\ $c \in \mathcal{C}$ and $
  ct\rightarrow u$ (i.e., $ct$ and $u$ are in the reaction relation). 
\end{enumerate}
\end{defi}

\noindent 
In the case of a reactive system defined on a category of contexts, a
state is a term $t$, and an associated label is a context $c$ such
that $ct$ reduces.  In the following, we will consider also lts's
obtained by reducing the set of transitions of the context lts. In the
sequel, we will use the word lts to refer to any such lts obtained
from a context lts.

Any lts induces a bisimilarity relation as follows:

\begin{defi}[Bisimilarity] Let $\utrans{c}$ be a lts. 
\begin{enumerate}[(i)]
\item A symmetric relation ${\mathcal R}\subseteq \bigcup_{I\in {\mathcal C}} {\mathcal C}(0,I) \times {\mathcal C}(0,I)$
on the states of the lts is a \emph{bisimulation} if:
\[\langle a,b \rangle \in {\mathcal R} \ \wedge \  a \utrans{f} a'\ \Longrightarrow \ \exists b' .\ 
b \utrans{f} b'\    \wedge \  \langle a',b' \rangle \in {\mathcal R}
\ .  
\]
\item We call \emph{bisimilarity}  the largest bisimulation.
\item The bisimilarity on the context lts  is called \emph{context bisimilarity} $\csim$.
\end{enumerate}
\end{defi}

\noindent
It is easy to check that the context bisimilarity is a \emph{congruence} w.r.t. all contexts, i.e.,
if $a \csim b$, then for any context $c$, $ca \csim cb$.
However, intuitively only those contexts which contain the \emph{minimal}
amount of information for a reaction to fire are relevant, while the others are redundant.
Moreover, often context bisimilarity gives an equivalence which is
 too coarse, as we will see also in this paper. Thus,  in \cite{LM00},  the authors proposed 
 a categorical criterion for identifying the ``smallest context allowing a reaction''. They defined 
 \emph{relative pushouts} (RPOs), of which \emph{idem relative pushouts} (IPOs) are a special case.
 One can define a lts using IPOs. Leifer-Milner's central result consists in showing that, under 
 a suitable categorical condition, such lts
 is well-behaved, in the sense that the induced bisimilarity is a congruence.
 
 \begin{defi}[RPO/IPO]
\hfill \begin{enumerate}[(i)]
\item
  Let $\mathcal{C}$ be a category and let us consider the commutative
  diagram in Fig.~\ref{runo}(i). Any tuple $\langle I_5, e,f,g
  \rangle$ which makes diagram in Fig.~\ref{runo}(ii) commute is
  called a \emph{candidate} for (i). A \emph{relative pushout (RPO)}
  is the smallest such candidate, i.e., it satisfies the universal
  property that given any other candidate $\langle I_6, e',f',g'
  \rangle$, there exists a \emph{unique} mediating morphism $h: I_5
  \rightarrow I_6$ such that both diagrams in Fig.~\ref{runo}(iii) and
  Fig.~\ref{runo}(iv) commute.
\item
  A commutative square such as diagram (i) in Fig~\ref{runo} is an
  \emph{idem pushout (IPO)} if $\langle I_4, c,d, \mathit{id}_{I_4}
  \rangle$ is its RPO.
\end{enumerate}
\end{defi}
 
 \begin{figure} 
 \[ 
 \xymatrix@R-15 pt@C-15 pt{ 
&&I_4 \cr\cr
  I_2  \ar[rruu]^{c}   &&&&  I_3  \ar[lluu]_{d}\cr\cr
&&0    \ar[lluu]^{t}   \ar[rruu]_{l} 
\cr 
 &&(i)
} 
\enspace
\xymatrix@R-15 pt@C-15 pt{ 
 && I_4 \cr\cr
 I_2 \ar[rr]^e \ar[rruu]^{c} && I_5 \ar[uu]_g &&  I_3 \ar[ll]_f \ar[lluu]_{d}
\cr\cr
 && 0 \ar[lluu]^{t} \ar[rruu]_{l}  
\cr
 &&(ii)
} 
\enspace
\xymatrix@R-15 pt@C-15 pt{ 
&& I_6 \cr\cr
 I_2 \ar[rr]^{e} \ar[rruu]^{e'}  
&& I_5 \ar[uu]_h&&  I_3 \ar[ll]_f \ar[lluu]_{f'}
\cr\cr
 &&{\mathstrut}
\cr
 &&(iii)
} 
\enspace
\xymatrix@R-15 pt@C-15 pt{ 
&& I_4 \cr\cr
 I_6 
  \ar[rruu]^{g'}  
 && I_5 \ar[ll]_h \ar[uu]_{g}
\cr\cr
 &&{\mathstrut}
\cr
 &&(iv)
}
\] 
 \caption{Redex Square and Relative Pushout.}
 \label{runo}
\end{figure}

\begin{defi}[IPO Transition System]\hfill
\hfill \begin{enumerate}[(1)]
\item States: arrows $t: 0\rightarrow I$ in $\mathcal{C}$, for any $I$;
\item Transitions: $t \itrans{c} dr $ \iff\ $d\in \mathcal{D}$, $ct = dl$, $\langle  l,r \rangle \in {\mathbf R}$ and the diagram in Fig.~\ref{runo}(i) is an IPO.
\end{enumerate}
\end{defi}

Let $\isim$ denote the bisimilarity induced by the IPO lts.

\begin{defi}[Redex Square]
Let ${\mathbf C}$ be a reactive system and $t:0\rightarrow I_2$ an arrow in ${\mathcal C}$.
A \emph{redex square} (see Fig.~\ref{runo}(i)) consists of a left-hand side $l:0 \rightarrow I_3$ of
a reaction rule  $\langle l:0 \rightarrow I_3, r:0 \rightarrow I_3\rangle \in {\mathbf R}$, a context
$c:I_2 \rightarrow I_4$ and a reactive context $d:I_3 \rightarrow I_4$ such that $ct =dl$.

A reactive  system $\mathbf{C}$ is said to \emph{have redex RPOs} if every redex square has an RPO.
 \end{defi} 

 \begin{figure} 
\[\xymatrix@R-15 pt@C-15 pt{ 
\ar[dddd]_{e_0} \ar[rrrr]^{f_0} &&&& \ar[dd]^{f_1} \cr\cr
     &&&& \ar[dd]^{e_2}\cr\cr
\ar[rr]_{g_0} && \ar[rr]_{g_1}&& \cr
&& (i)
}
\qquad\qquad
\xymatrix@R-15 pt@C-15 pt{ 
\ar[dd]_{e_0} \ar[rr]^{f_0} && \ar[rr]^{f_1}\ar[dd]^{e_1} && \ar[dd]^{e_2}   
\cr\cr
\ar[rr]_{g_0} && \ar[rr]_{g_1} &&
\cr\cr\cr
&& (ii)
}
\] 
 \caption{IPO pasting.}
 \label{rdue}
\end{figure}

The following is a fundamental lemma stating  a property of IPO squares.

\begin{lem}[IPO pasting, \cite{LM00}]\label{ipo-pasting}
Suppose that the square in Fig. \ref{rdue}(i) has an RPO and that  both squares in Fig. \ref{rdue}(ii) commute.
\hfill \begin{enumerate}[\em(i)]
\item If the two squares of Fig. \ref{rdue}(ii) are IPOs so is the outer rectangle.
\item It the outer rectangle and the left square of Fig. \ref{rdue}(ii) are IPOs so is the right square. 
\end{enumerate}
\end{lem}

\noindent
From the above lemma Leifer and Milner derived their central result:

\begin{thm}[\cite{LM00}]\label{ipocon}
Let $\mathbf{C}$ be a reactive system having redex RPOs. Then the IPO
bisimilarity $\isim$ is a congruence w.r.t. all contexts, i.e., if
$a\isim b$ then for all $c$ of the appropriate type, $ca\isim cb$.
\end{thm}

\subsection{Weak  Bisimilarity}
For dealing with the $\lambda$-calculus, it will be useful to consider the weak versions of the context and IPO lts's defined above, together with the corresponding notions of \emph{weak bisimilarities}.

One can proceed in general, by defining a weak lts from a given lts:

\begin{defi}[Weak lts and Bisimilarity] \label{wIPO}
Let $\utrans{\alpha}$ be a lts, and let $\tau$ be a label (identifying an unobservable action).
\hfill \begin{enumerate}[(i)]
\item We define the \emph{weak lts} $\wutrans{\alpha}$ by 
\[t \wutrans{\alpha} u 
  \quad\hbox{iff}\quad
  \begin{cases} t  \utrans{\tau}^* u & 
 \mbox{ if } \alpha=\tau
 \\  t \utrans{\tau}^*  t'  \utrans{\alpha} u'  \utrans{\tau}^* u &
  \mbox{ otherwise ,}  \end{cases}
\]
where $\utrans{\tau}^*$ denotes the reflexive and transitive closure of 
 $\utrans{\tau}$.
\item Let us call \emph{weak bisimilarity} the  bisimilarity induced by the 
 weak lts.
\end{enumerate}
\end{defi}

\noindent 
The above definition differs from the one proposed in \cite{LM00}, where, in case $\alpha\neq \tau$, 
$\wutrans{\alpha} $ is defined by $\utrans{\alpha}\circ \utrans{\tau}^*$. We cannot use the latter, since it
discriminates $\lambda$-terms which are equivalent in the usual semantics.

The following easy lemma gives a useful characterization of the weak bisimilarity,
whereby any $\utrans{\alpha}$-transition is mimicked by a   $\wutrans{\alpha}$-transition:

\begin{lem}\label{altcar}
Let $\utrans{\alpha}$ be a lts and let  $\wutrans{\alpha}$ be the corresponding weak lts.
The induced weak bisimilarity is the greatest symmetric relation $ {\mathcal R}$ s.t.:
\[\langle a,b \rangle \in {\mathcal R} \ \wedge \  a \utrans{f} 
a'\ \Longrightarrow \ \exists b' .\ 
b \wutrans{f} b'\    \wedge \  \langle a',b' \rangle \in {\mathcal R}
\ .
\]
\end{lem}

The following lemma provides a coinduction ``up-to'' principle, which will be useful in the sequel: 

\begin{lem}\label{upto}
Let $\utrans{\alpha}$ be a lts and let  $\wutrans{\alpha}$ be the corresponding weak lts.
The induced weak bisimilarity is the greatest symmetric relation $ {\mathcal R}$ s.t.:
\[\langle a,b \rangle \in {\mathcal R} \ \wedge \  a \stackrel{f}{\Longrightarrow}\!\! {}'
a'\ \Longrightarrow \ \exists b' .\ 
b \wutrans{f} b'\    \wedge \  \langle a',b' \rangle \in {\mathcal
R}^* \ ,
\]
 where  
$\stackrel{f}{\Longrightarrow}\!\! {}'$ denotes $ \utrans{\tau}^* \circ \stackrel{f}{\longrightarrow}$ ($f$
 is possibly $\tau$), and ${\mathcal R}^*$ denotes the reflexive and transitive closure of 
${\mathcal R}$.
\end{lem}
\begin{proof} Let us call ``bisimulation up-to'' a relation 
${\mathcal R}$ as in the statement of the lemma.  In order to prove
the claim, it is sufficient to prove that, if ${\mathcal R}$ is a
bisimulation up-to, then ${\mathcal R}^*$ is a bisimulation. Let
${\mathcal R}$ be a bisimulation up-to. First, one can easily check
that $(a {\mathcal R}^* b \ \wedge \ a \Longrightarrow a') \
\Longrightarrow\ \exists b'.\ (b\Longrightarrow b' \ \wedge\ a'
{\mathcal R}^*b')$ (by induction on the length of the chain $a\
{\mathcal R} \ldots {\mathcal R}\ b$). Now, let $a = a_0 \ {\mathcal
R}\ a_1 \ldots a_{n-1} \ {\mathcal R} \ a_n = b$ and $a\,
\smash{\stackrel{f}{\Longrightarrow}}\, a'$. We prove that $\exists b'.\
(b\,\smash{\stackrel{f}{\Longrightarrow}}\, b' \ \wedge\ a' {\mathcal
R}^*b')$, by induction on $n\geq 0$.  If $n=0$, the claim is
immediate. If $n>0$ and $a\, \smash{\stackrel{f}{\Longrightarrow}}
{}'\, a'' \Longrightarrow a'$, then, since ${\mathcal R}$ is a
bisimulation up-to, $a_1\, \smash{\stackrel{f}{\Longrightarrow}}\, a''_1\
\wedge \ a'' {\mathcal R}^* a''_1,$ and, by what we have proved
before, $\exists a'_1.\ ( a''_1 \Longrightarrow a'_1 \ \wedge\ a'
{\mathcal R}^* a'_1)$. Finally, by induction hypothesis, $\exists b'.\
(b\, \smash{\stackrel{f}{\Longrightarrow}}\, b' \ \wedge\ a'_1 {\mathcal
R}^* b' )$. Hence $a' {\mathcal R}^* b'$.
\end{proof}

For dealing with the $\lambda$-calculus, we will consider a notion of \emph{weak IPO bisimilarity}, where the identity
context is unobservable. Such notions of weak IPO bisimilarities are
 not  congruences w.r.t. all contexts, in general, however, as observed in 
 \cite{LM00} (end of Section~5), they are congruences at least w.r.t. reactive contexts:

\begin{thm}\label{wcon}
Let  $\mathbf{C}$ be a reactive system having redex RPOs. Then the
weak IPO bisimilarity $\wisim$, where the identity
context is unobservable, is a congruence w.r.t. reactive contexts.
\end{thm}

\section{Extending the Theory of Reactive Systems}\label{ext}

In this section, we present some original results concerning the lts obtained by the RPO construction. These results concern two issues:

\begin{enumerate}[\hbox to8 pt{\hfill}]
 \item{\hskip-12 pt\bf Weak-bisimilarity:}\  Since in the $\lambda$-calculus the weak bisimilarity is the equivalence to be used, we present some general conditions assuring that the weak bisimilarity, on the lts obtained by an IPO construction, is a congruence w.r.t. all contexts.
\item{\hskip-12 pt\bf Pruning the lts tree:}\ In order to obtain a feasible lts, i.e., a lts with a reduced set of transitions, possibly finitely branching, it is often necessary to prune the lts obtained by an IPO construction. We present some general conditions allowing to prune IPO lts, without modifying the induced (weak)-bisimilarity.
\end{enumerate}
We present our results in two different versions, the first one is
quite simple, but it does not apply to our particular case, so we
present a second version that is more involved but suits our needs.
We choose to present the simple first version of our results as an
introduction to the second one, and also because it can have
applications in modeling languages different from the
$\lambda$-calculus.

Some preliminary definitions are necessary.

\begin{defi}
Given a lts obtained by the IPO construction: 
\begin{enumerate}[$\bullet$]
\item
Given a set of labels $L$, the \emph{$L$-restricted IPO lts} is the
 lts obtained by removing from the IPO lts all transitions not labeled
 by elements in $L$.  We denote by $\wsim_{L}$ the weak bisimilarity
 induced by the $L$-restricted IPO lts.
\item
We denote by $R$ the set of labels that are reactive contexts.  We
denote by $\wrsim$ the weak bisimilarity induced by the $R$-restricted
IPO lts.
\item
In a reactive system, we say that the family of IPO transitions with
label $f: I_0 \rightarrow I_1$ is \emph{definable by contexts} if
there exists a list of contexts $e_1, \ldots , e_h: I_0 \rightarrow
I_1$ such that, for all $t: 0 \rightarrow I_0$, we have that: $\forall
i.\ t \itrans{f} e_i t$ and $t \itrans{f} t' \ \Longrightarrow \
\exists i. \ t' = e_i t$.
\end{enumerate}
\end{defi}

\noindent 
Intuitively, a family of IPO transitions with label $f: I_0
\rightarrow I_1$ is definable by contexts if $f$ is an IPO for any
arrow $t: 0 \rightarrow I_0$ and the IPO transitions on $f$ can be
described by contexts, that is, they do not modify the internal
structure of the term $t$.

\begin{prop} \label{weak-congruence}
Let  $\mathbf{C}$ be a reactive system having redex RPOs. If any IPO context is either reactive or definable by contexts (or both), then the weak IPO bisimilarity $\wisim$ (with the identity
IPO context unobservable) is a congruence.  Moreover $\wisim$ coincides with $\wrsim$.
\end{prop}
\proof
Consider the relation $S = \{\, \langle ct, cu \rangle  \mid t \wrsim u, \ c \ \mbox{context}\, \}$. 
It is immediate that $\wisim \subseteq \wrsim$, and from this, $\wisim \subseteq \{ \langle ct, cu \rangle  \mid t \wisim u, \ c \ \mbox{context} \} \subseteq S$.
If we prove also the inclusion $S \subseteq \wisim$, then all relations are equal and $\wisim$ coincides with its contextual closure, i.e., it is congruence. 
By Lemma~\ref{altcar}, in order to prove $S \subseteq \wisim$ it is sufficient to show that, for any $ \langle ct, cu \rangle  \in S$, if $ct \itrans{f} t'$ then there exists $u'$ s.t.\ $cu \witrans{f} u'$ with $t' S u'$.

Consider the following diagram:
\[
\xymatrix{
0 \ar[r]^{t} \ar[d]_{l} & I_0 \ar[r]^{c} \ar[d]_{f'} & I_2 \ar[d]^{f} \\
I_3 \ar[r]_{d}          & I_1 \ar[r]_{d'}           & I_4 
}
\] 
where the outermost rectangle is the IPO inducing the transition $ct \itrans{f} t'$, 
namely $t' = d' d r$ with $ \langle l,r \rangle $ a reaction rule, while the left square is a RPO  of the redex square.  
By Lemma~\ref{ipo-pasting}, the IPO pasting property, we have that also the  right-hand square of the diagram is an IPO.
 
There are two cases to consider: 
\begin{enumerate}[(i)]
\item
If the context $f'$ is definable by contexts, since $t \itrans{f'}
dr$, there exists a context $e$ such that $d r = e t$ and $t' =d' e
t$, it follows that $u \itrans{f'} e u$.  That is, there exist a
reaction rule $ \langle l_1, r_1 \rangle $ and a reactive context
$d_1$ s.t.\ $e u = d_1 r_1$, and the left-hand square of the following
diagram is a IPO.
\[
\xymatrix{
0 \ar[r]^{u} \ar[d]_{l_1} & I_0 \ar[r]^{c} \ar[d]_{f'} & I_2 \ar[d]^{f} \\
I_3 \ar[r]_{d_1}          & I_1 \ar[r]_{d'}           & I_4 
}
\] 
Since the right-hand square is IPO, by the IPO pasting property,
Lemma~\ref{ipo-pasting}, also the outermost rectangle is an IPO.  It
follows that $c u \itrans{f} d' d_1 r_1 = d' e u$, which implies the
claim.

\item
If the context $f'$ is reactive, then it so also the context $d' f'$
(composition of reactive contexts) and the context $c$ (reactive
contexts are composition-reflecting).  Moreover, by the definition of
bisimilarity, there exists $u_0$ such that $u \witrans{f'} u_0$ (which
means
$u \itrans{Id}^{\ast} u_1 \itrans{f'} u_2 \itrans{Id}^{\ast} u_0$)
with $u_0 \wrsim dr$.  Since $c$ is reactive and squares of the form
\[
\xymatrix{
I_0 \ar[r]^{c} \ar[d]_{Id} & I_2 \ar[d]^{Id} \\
I_1 \ar[r]_{c}           & I_3 
}
\] 
are IPOs, by composition of IPO squares (and by induction) it is easy
to prove that $cu \itrans{Id}^{\ast} cu_1 \itrans{f} d'u_2
\itrans{Id}^{\ast} d'u_0$, which implies the claim.\qed
\end{enumerate}

For dealing with the $\lambda$-calculus, we present a second
proposition that is similar in spirit to
Proposition~\ref{weak-congruence}, although it is not a direct
generalization.  The second proposition considers both the category of
unary linear term contexts and a category of ``multi-holed'' linear
term contexts.  The category of unary contexts is the most suitable
for the IPO construction, while the category of multi-holed contexts
is useful to represent some transitions (in the lts) through
insertions of terms in suitable contexts.

The following definition formalizes the relation existing between the
two categories of contexts.

\begin{defi}
A category $\mathcal{D}$ is a \emph{list extension} of a category $\mathcal{C}$ if the following hold:
\begin{enumerate}[$\bullet$]
\item $\mathcal{C}$ contains a distinguished object $0$.
\item The objects of $\mathcal{D}$ are finite  lists of objects of $C$ different from $0$.
\item By identifying $0$ with the empty list $\langle \ \rangle$, and any other object $I$ in $\mathcal{C}$ with the singleton list $\langle I \rangle$, $\mathcal{C}$ is a \emph{full subcategory} of $\mathcal{D}$.
\item There exists a \emph{concatenation} functor $\otimes$ from $\mathcal{D} \times \mathcal{D}$ to $\mathcal{D}$ acting as concatenation on objects $\langle I_0, \ldots, I_n \rangle \otimes \langle J_0, \ldots, J_m \rangle = \langle I_0, \ldots, I_n, J_0, \ldots, J_m \rangle $ and being associative on arrows. 
\end{enumerate} 
\end{defi}
\noindent 
In the spirit of the previous remark we will call unary (single-holed) contexts the arrows in  $\mathcal{C}$ (with domain different from $0$), and multi-holed contexts the arrows in $\mathcal{D}$.

Two other definitions are necessary.
\begin{defi}
Given a reactive system $\mathbf{C}$ on a category $\mathcal{C}$, and a category $\mathcal{D}$, list extension of $\mathcal{C}$:
\begin{enumerate}[(i)]
\item
we define a multi-holed context $g : \langle I_0, \ldots, I_n \rangle \rightarrow  I $ \emph{IPO uniform} if for any  context $f: I \rightarrow J$
 appearing as label in the IPO lts, there exists a list of multi-holed contexts  $g_1: \langle I_{1,0}, \ldots, I_{1,n_1} \rangle \rightarrow  
 J,  \ldots,  
g_h : \langle I_{h,0}, \ldots, I_{h,n_h} \rangle \rightarrow J $, and a list of functions 
$l_1 : \{0, \ldots, n_1 \} \rightarrow \{0, \ldots, n \},
\ldots, 
 l_h : \{0, \ldots, n_h \} \rightarrow \{0, \ldots, n \}$   such that,
for any n-tuple of $\mathcal{C}$ terms $t_0 : 0 \rightarrow I_0, \ldots, t_n : 0 \rightarrow I_n$, we have that:
\begin{enumerate}[$- $]
\item $\forall i.\ g(t_0 \otimes \ldots \otimes t_n) \itrans{f} g_i (t_{l_i(0)} \otimes \ldots \otimes t_{l_i(n_i)})$ and
\item
 $g(t_0 \otimes \ldots \otimes t_n) \itrans{f} t' \Longrightarrow \exists i. \ t' = g_i (t_{l_i(0)} \otimes \ldots \otimes t_{l_i(n_i)})$;
 \end{enumerate}
\item[(ii)]
a context $g : \langle I_0, \ldots, I_n \rangle \rightarrow  I $ has a \emph{reactive index $i$} if for any list of $n$ terms $t_0 : 0 \rightarrow I_0, \ldots, t_{i-1}, t_{i+1}, \ldots,  t_n : 0 \rightarrow I_n$, the context $g (t_0 \otimes \ldots \otimes t_{i-1} \otimes id_{I_i} \otimes t_{i+1} \otimes \ldots \otimes t_n): I_i \rightarrow I$, seen as a context in $\mathcal{C}$, is reactive. 
\end{enumerate} 
\end{defi}
\noindent 
Intuitively, a context $g$ is IPO uniform if the behavior wrt the IPO reaction of the term $g( t_{l_i(0)} \otimes \ldots \otimes t_{l_i(n_i)} )$ does not depend on the terms $t_{l_i(0)}, \ldots, t_{l_i(n_i)}$.  We remark that the notion of ``uniform'' is not a generalization of the notion of ``definable by contexts''.
\begin{prop}\label{list-weak-congruence}
Let $\mathbf{C}$ be a reactive system  having redex RPOs. 
\begin{enumerate}[\em(i)]
\item
The weak IPO bisimilarity $\wisim$ (with the identity IPO context unobservable) is a congruence if there exists a category $\mathcal{D}$, list extension of $\mathcal{C}$ such that any (multi-holed) context $g : \langle I_0, \ldots, I_n \rangle \rightarrow  I $ is either IPO uniform or it has a reactive index (or both).
\item
Moreover, if the reaction relation is deterministic, i.e., any term can react in at most  one possible way, then the relation $\wisim$ coincides with $\wrsim$.
\end{enumerate}
\end{prop}

\proof Here we present only the proof of point (ii).  The proof of
point (i) is almost identical and can be derived, from the present
proof, by substituting the relation $\wrsim$ with $\wisim$, and by
simplifying some steps.

  By repeating the same arguments used at the beginning of the proof
  of Proposition~\ref{weak-congruence}, it is sufficient to prove that
  the relation
\[S = \{\, \langle\,g(t_0 \otimes \ldots \otimes t_n), g(u_0
\otimes \ldots \otimes u_n)\,\rangle\,\mid\, g : \langle I_0, \ldots,
I_n \rangle \rightarrow I, \ \forall i.\, t_i{}\wrsim{}u_i\, \}
\]
  is contained in the weak bisimilarity.  By Lemma~\ref{upto}, it is
  sufficient to show that for any $ \langle g(t_0 \otimes \ldots
  \otimes t_n), g(u_0 \otimes \ldots \otimes u_n) \rangle \in S$ and
  IPO-transition $f$, if $g(t_0 \otimes \ldots \otimes t_n)
  \witrans{f} t$, with $\itrans{f}$ the last step of the chain of
  reactions, then there exists $u$ s.t.\ $ g(u_0 \otimes \ldots
  \otimes u_n) \witrans{f} u$ with $t S^{\ast} u$.  The proof is by
  double induction on the number of steps of the transition $g(t_0
  \otimes \ldots \otimes t_n) \witrans{f} t$, and on the number $n$ of
  holes in the list context $g$.

The basic case is when $g(t_0 \otimes \ldots \otimes t_n) \witrans{f} t$ in $0$ steps ($f = id$), in this case there is nothing to prove. 

Now suppose $g(t_0 \otimes \ldots \otimes t_n) \itrans{f'} t' \witrans{f''} t$, in this case $(f' = id \wedge f'' = f)$ or  $(f' = f \wedge f'' = id \wedge t' = t$),

There are two cases to consider: 
\begin{enumerate}[(i)]
\item
The context $g$ is IPO-uniform: in this case there exists a context
$e: \langle I'_{0}, \ldots, I'_{n'} \rangle \rightarrow J_1$ and a
function $l : \{0, \ldots n' \} \rightarrow \{0, \ldots , n \}$ such
that $t'= e(t_{l(0)} \otimes \ldots \otimes t_{l(n')})$ and $g(u_{0}
\otimes \ldots \otimes u_{n}) \itrans{f'} e(u_{l(0)} \otimes \ldots
\otimes u_{l(n')})$.  By application of the inductive hypothesis, on a
smaller number of transitions steps, there exists $u$ s.t.\ $e(u_{l(0)}
\otimes \ldots \otimes u_{l(n')}) \witrans{f''} u$ with $t S^{\ast}
u$, and from which the claim follows.
\item
The context $g$ has a reactive index $i$, for the sake of simplicity,
assume $i = 0$. Consider the arrow $g' = g (t_0 \otimes id_{I_{1}}
\otimes \ldots \otimes id_{I_{n}}) : \langle I_1, \ldots, I_n \rangle
\rightarrow I $.  Since $g'(t_{1} \otimes \ldots \otimes t_n) = g(t_0
\otimes \ldots \otimes t_n) \witrans{f} t$, by inductive hypothesis,
on the number of holes in the multi-holed contexts, there exists $u$
such that $g'(u_{1} \otimes \ldots \otimes u_n) = g(t_0 \otimes u_{1}
\otimes \ldots \otimes u_n) \witrans{f} u$, with $t S^{\ast} u$.

Now consider the context $g'' = g(Id_{I_0} \otimes u_{1} \otimes \ldots \otimes u_n) : I_o \rightarrow I$. The context $g''$  is reactive and $g''(t_0) \witrans{f} u$.  
To obtain the claim, it remains to prove that there exists $u'$ s.t.\ $g''(u_0) = g(u_0 \otimes \ldots \otimes u_n) \witrans{f} u'$, with $u S^{\ast} u'$.

More generally we prove that for any reactive context $g_o : J_0 \rightarrow J_1$, any IPO context $f: J_1 \rightarrow J_2$, and any pair of terms $t_o, u_o$,  if $t_o \wrsim u_o$ and $g_o(t_o)\witrans{f} t_o'$  then there exists  $u_o'$ s.t.\ $g_o(u_o) \witrans{f} u_o'$ and $t_o' S^{\ast} u_o'$. 
The proof is by induction on the number of steps in the transition $g''(t_0) \witrans{f} t_o'$.
The basic case is when  the reaction is of zero steps; in this case there is nothing to prove. 

For the inductive case consider the following diagram of IPO squares defining the first reaction in the chain
\[
\xymatrix{
0 \ar[r]^{t_o} \ar[d]_{l} & J_0 \ar[r]^{g_o} \ar[d]_{f''} & J_1 \ar[d]^{f'} \\
J_3 \ar[r]_{d}          & J_4 \ar[r]_{d'}           & J_5 
}
\] 
We need to consider two cases.
The first one is where $f'$ is a reactive context ($f' \in \{ f, Id \}$). Since reactive contexts are 
composition-reflecting, then also the IPO context $f''$ is reactive.  By the definition of bisimilarity, $u_o \witrans{f''} u_i$ with $u_i \wrsim d r$.  By reactivity of $g_o$, using suitable IPO pasting diagrams, we can prove $g_o(u_o) \witrans{f'} d' u_i$. Now by applying the inductive hypothesis to the reduction $d'( d r) \witrans{f} t_o'$, we obtain the claim. 

The second case is where  $f'$ is a non reactive context ($f' = f$).  Since reactive contexts are compositional reflecting, then also the IPO context $f''$ is non reactive and therefore, by hypothesis, IPO uniform. Notice that the context $Id $ is an IPO context for the term $f''(t_o)$, by the IPO uniformity of $f''$, $Id $ is an IPO context also for $f''(u_o)$ and there exist a list context $g' \langle J_1, \ldots J_1 \rangle \rightarrow J$ s.t.\  $t_o \itrans{f''} g'(t_o, \ldots, t_o)$ and also  $u_o \itrans{f''} g'(u_o, \ldots, u_o)$.  Notice that,  if the reduction relation is deterministic, two terms that reduce one to the other via $\tau$ transitions are weakly bisimilar. It follows that $g_o(t_o)  \itrans{f} d' g'(t_0, \ldots, t_0) \wrsim t_o' $ and, by IPO pasting, $g_o(u_o) \itrans{f} d' g'(u_o, \ldots, u_o)$, from which we derive the claim.\qed
\end{enumerate}

\begin{rem}
Propositions~\ref{weak-congruence} and ~\ref{list-weak-congruence} above, about 
 congruence of the weak IPO  bisimilarity, are more related than what they look at first glance. From one side, 
  by exploiting the fact that the composition of a non-reactive context with any context gives a non-reactive 
 context, one can show that, if the non-reactive IPOs are definable by contexts,  then any non-reactive context is IPO-uniform.
  Note that the condition of ``definability by context'' is in general simpler to verify than the one of ``IPO-uniformity'', and so we prefer to present the given formulation of Proposition~\ref{weak-congruence}. On the other side, it would be possible to extend the notion of ``definability by context'' to the case of list extension categories, however to this
  aim  it would be necessary to present a series of new definitions, necessary to lift the IPO construction to the list extension categories. For the sake of simplicity, we prefer to avoid the introduction of these further notions.  
\end{rem}

\section{The Lambda Calculus}\label{lamc}

First, we recall the $\lambda$-calculus syntax together with
\emph{lazy} and \emph{cbv} reduction strategies and observational
equivalences. Then, we show how to apply the RPO technique to
$\lambda$-calculus, viewed as a context category, and we discuss some
problematic issues.

\subsection{Syntax, Reduction Strategies, Observational Equivalences}

\begin{defi}[Syntax] The set of \emph{$\lambda$-terms} $\Lambda$ is defined by
\[ (\Lambda\ni)\ M::=\  x \  | \ MM\ | \ \lambda x. M \ , \]
where $x\in \mathit{Var}$ is an infinite set of variables.
Let $\mathit{FV}(M)$ denote the set of free variables in $M$, and
let us denote by $\Lambda^0$ the set of \emph{closed} $\lambda$-terms.
\end{defi}  

As usual, $\lambda$-terms are taken up-to $\alpha$-conversion, and
application associates to the left.  We consider the standard notions
of $\beta$-rule and $\beta_V$-rule:
 
 \begin{defi}
\hfill \begin{enumerate}[(i)]
\item $\beta$-rule: $(\lambda x.M)N \rightarrow_{\beta} M[N/x]$;
\item $\beta_V$-rule: $(\lambda x.M)N \rightarrow_{\beta_V}  M[N/x]$,
  if $N$ is a variable or a  $\lambda$-abstraction.
\end{enumerate}
 \end{defi}

\noindent 
As usual, we denote by $=_{\beta}$ and $=_{\beta_V}$ the corresponding conversions.
 
A \emph{reduction strategy}
on the $\lambda$-calculus determines, for each term which is not a value, a suitable $\beta$-redex
appearing in it to be contracted. 
The lazy and cbv reduction strategies are defined on closed $\lambda$-terms as follows:

\begin{defi}[Reduction Strategies]\label{strat}
\hfill \begin{enumerate}[(i)]
\item The \emph{lazy strategy} $\rightarrow_l \subseteq \Lambda^0 \times \Lambda^0$  reduces the leftmost
$\beta$-redex, not appearing within a $\lambda$-abstraction. Formally, $\rightarrow_l$ is defined
by the rules:
\[\reguno{}{(\lambda x.M)N  \rightarrow_l M[N/x] } \qquad
\reguno{N\rightarrow_l N' }{NP\rightarrow_l N'P }
\]
\item The \emph{call by value strategy} $\rightarrow_v \subseteq \Lambda^0 \times \Lambda^0$  reduces the leftmost
$\beta_V$-redex, not appearing within a $\lambda$-abstraction. Formally, $\rightarrow_v$ is defined
by the following rules:  
\[\reguno{}{(\lambda x.M)V  \rightarrow_v M[V/x] } \qquad
  \reguno{N\rightarrow_v N' }{NP\rightarrow_v N'P }\qquad
  \reguno{N\rightarrow_v N'}{(\lambda x.M)N  
  \rightarrow_v(\lambda x.M) N' } 
\]
where $V$ is a closed value, i.e., a 
$\lambda$-abstraction.
\end{enumerate}
\end{defi}

\noindent 
We denote by $\rightarrow_{\sigma}^{*}$ the reflexive and transitive closure of a strategy
$\rightarrow_{\sigma}$, for $\sigma \in \{ l,v\}$, by $\mathit{Val}_{\sigma}$ the set of values, i.e., the set of terms on which the reduction strategy halts (which coincides with the set of $\lambda$-abstractions
in both cases), 
and by $M\Downarrow_{\sigma}$ the fact that there exists $V\in \mathit{Val}_{\sigma}$ such that
$M \rightarrow_{\sigma}^{*}V$.

As we will see in Section~\ref{lclt} below, each strategy defines a (deterministic) reactive system on
$\lambda$-terms in the sense of Definition~\ref{rslm}.  To this aim, it is useful to notice that the 
above reduction strategies can be alternatively determined by specifying suitable  sets of
\emph{reactive contexts} (see Remark~\ref{strate} below), 
which are subsets of the following  \emph{unary contexts}, i.e.,
contexts with a single hole:

 \begin{defi}[Unary Contexts]  Let $P\in \Lambda$. The \emph{unary contexts} are:
\[ C[\ ] \ ::= \  [\  ] \ | \ PC[\ ] \ | \ C[\ ] P \ |\ \lambda x.C[\ ]\ .\]  
The \emph{closed unary contexts} are the unary contexts with no free variables.
 \end{defi}
 
\begin{rem} \label{strate}
\hfill \begin{enumerate}[(i)]
\item The lazy strategy  $\rightarrow_l $ is the closure of the $\beta$-rule under the reactive contexts, corresponding to the closed applicative contexts: \ $ D[\ ] \ ::= \ [\  ]\  |\ D[\ ] P \ , $ where $P\in \Lambda^0$. 
\item
The cbv strategy $\rightarrow_v$ is the closure of the $\beta_V$-rule
 under the following closed reactive contexts: \ $ D[\ ] \ ::= \ [\  ]\ |\   D[\ ] P     \ | \  (\lambda x. M)D[\ ] \ ,$
 where $P , \lambda x.M\in \Lambda^0$. 
\end{enumerate}
 \end{rem} 

\noindent 
Each strategy induces an \emph{observational (contextual) equivalence}
\`a la Morris on closed terms, when we consider programs as {\it black
boxes} and only observe their ``halting properties''.

\begin{defi}[$\sigma$-observational Equivalence]  
Let $\rightarrow_{\sigma}$ be a reduction strategy and let $M,N\in
\Lambda^0$.  The {\em observational equivalence} $\approx_{\sigma}$ is
defined by
\[ M\approx_{\sigma} N\quad \mbox{\iff\  for any closed unary context } C[\ ].\ 
C[M]\Downarrow_{\sigma}\Leftrightarrow C[N]\Downarrow_{\sigma}\ .
\]
\end{defi} 
The definition of $\approx_{\sigma}$ can be extended to open terms by
considering closing (by-value) substitutions, i.e., for $M,N \in
\Lambda$ s.t.\ $FV(M,N) \subseteq \{ x_1 , \ldots , x_n \}$, we define:
\[ M \widehat{\approx}_{\sigma}N\quad\mbox{\iff\ for all closing (by-value) substitutions $\vec{P}$}, M[\vec{P}/\vec{x}] 
\approx_{\sigma} N[\vec{P}/\vec{x}] \ .
\]

\begin{rem}
Often in the literature, the observational equivalence is defined by
considering multi-holed contexts. However, it is easy to see that the
two notions of observational equivalences, obtained by considering
just unary or all multi-holed contexts, coincide.
\end{rem}

The problem of reducing the set of contexts in which we need to check
the behavior of two terms has been widely studied in the literature.
In particular, for both strategies in Definition~\ref{strat} above, a
\emph{Context Lemma} holds, which allows us to restrict ourselves to
\emph{applicative contexts} of the shape $ [\ ] \vec{P} $ ($[\ ]
\vec{V}$), where $\vec{P}$ ($\vec{V}$) denotes a list of closed terms
(values). Let us denote by $\approx_{\sigma}^{app}$ the observational
equivalence which checks the behavior of terms only in applicative
(by-value) contexts. This admits a coinductive characterization as
follows:

\begin{defi}[Applicative $\sigma$-bisimilarity]
\hfill \begin{enumerate}[(i)]
\item A relation ${\mathcal R} \subseteq \Lambda^0 \times \Lambda^0$ is 
\begin{enumerate}[$-$]
\item an \emph{applicative 
lazy bisimulation} if the following holds:
\[\langle M,N \rangle \in {\mathcal R} \ \Longrightarrow \ (M\Downarrow_{l}\
\Leftrightarrow\
N\Downarrow_{l}) \ \wedge\ 
\forall P\in \Lambda^0.\  \langle MP, NP\rangle\in
{\mathcal R} .
\]
\item an \emph{applicative 
cbv bisimulation} if the following holds:
\[\langle M,N \rangle \in {\mathcal R} \ \Longrightarrow \ (M\Downarrow_{v}\
\Leftrightarrow\
N\Downarrow_{v}) \ \wedge\ 
\forall V \mbox{ closed value}.\  \langle MV, NV\rangle\in
{\mathcal R} .
\]
\end{enumerate} 
\item The applicative equivalence $\approx_{\sigma}^{app}$ is the largest applicative bisimulation.
\end{enumerate}
\end{defi}

\noindent 
The following is a well-known result \cite{AO93,EHR92}:

\begin{lem}[Context Lemma]
\label{cont}
 $\approx_{\sigma} = \approx_{\sigma}^{app}$.
 \end{lem}
 
 By the Context Lemma, the class of contexts in which we have to check the behavior of terms 
 is smaller, however it is still infinite, thus the applicative bisimilarity is infinitely branching.
In the following,  we will study alternative coinductive characterizations
of the observational equivalences, arising from the application of Leifer-Milner technique.

\subsection{Lambda Calculus as a Reactive System}\label{lclt}

Both lazy and cbv $\lambda$-calculus can be endowed with a structure of
reactive system in the sense of Definition~\ref{rslm}, by considering
corresponding  context categories. 
 
\begin{defi}[Lazy, cbv $\lambda$-reactive Systems]
${\bf C}_{\sigma}^{\lambda}$, for $\sigma \in \{ l,v\}$,  consists of 
\begin{enumerate}[$\bullet$]
\item the category  whose objects are $0,1$, where the morphisms from 
0 to 1 are the closed terms  (up-to $\alpha$-equivalence), the morphisms from 1 to 1 are the  
unary closed contexts  (up-to $\alpha$-equivalence), and composition is context insertion;
\item the subcategory of reactive contexts is determined by the reactive contexts 
for the lazy and cbv strategy, respectively,  presented in Remark~\ref{strate};
\item the (infinitely many) reaction rules  are $(\lambda x.M)N\rightarrow_{\beta_{\sigma}}  M[N/x]$, for all $M, N$, where
$\rightarrow_{\beta_l}$ is $\rightarrow_{\beta}$-rule, while $\rightarrow_{\beta_v}$ is $\rightarrow_{\beta_V}$-rule.
\end{enumerate}
\end{defi}
\noindent
The above definition is well-posed, in particular the subcategory of reactive contexts is
composition-reflecting. 

One can easily check that the reactive system  $\mathbf{C}_{\sigma}^{\lambda}$ 
has redex RPOs;  this fact can be proved by rephrasing the corresponding proof for the category of term contexts of  \cite{Sew02}.  Here it is essential the fact that we consider only closed terms and closed contexts.

\begin{lem}
The reactive system ${\mathbf C}^{\lambda}_{\sigma}$, for $\sigma \in \{ l,v \}$, has redex RPOs.
\end{lem}
 
 The IPO contexts of a  closed term for the lazy and cbv reactive systems
are summarized in the second columns of the tables in Fig.~\ref{tb}.  Intuitively, such contexts are minimal for the given reduction 
to fire. Vice versa, contexts different from the ones above are not IPO;  \emph{e.g.} $C[\ ]P$, for terms of the shape $\lambda x.M$, is
not IPO if $C[\ ]$ is different from $\lambda x. C_1[\ ]$ and  $[\ ]$, because otherwise the reduction can fire already in $C[\ ]$.

\begin{figure} 
\begin{center}
\begin{tabular}{| l | l | l |}
  \multicolumn{3}{c}{\bf Lazy IPO  lts's}\\
\hline
term & IPO contexts & reactive IPO contexts\\
\hline\hline
 $\lambda x.M$  &  $[\ ]P,$ $(\lambda x. C[\ ]) P$, $PC[ \ ]$ \hspace*{0.2cm} 
 & [\  ]P \\
\hline
 $(\lambda x.M)N\vec{P}$ \hspace*{0.2cm}& $[\ ] $, $(\lambda x. C[\ ]) P$, $P C[ \ ]$
 & [\ ]  \\
\hline
\end{tabular} 
\vspace*{0.4cm}

\begin{tabular}{| l |l | l|}
\multicolumn{3}{c}{\bf Cbv IPO lts's}\\
\hline
term $M$ & IPO contexts & reactive IPO contexts \\
\hline\hline
 $\lambda x.M_1$  &  $[\ ]P,$ $(\lambda x. C[\ ]) P$, $RC[ \ ]$,  $(\lambda x.Q)C_1 [\ ]$\hspace*{0.2cm}
 &   $[\ ]P$, $(\lambda x.Q) [\ ]$ 
  \\
\hline
 $(\lambda x.M_1 )N\vec{P}$ \hspace*{0.2cm} & $[\ ] $, $(\lambda x. C[\ ]) P$, $RC[ \ ]$,  $(\lambda x.Q)C_1 [\ ]$&
 $[\ ] $
  \\
\hline
\end{tabular} \bigskip
\\ where $R$ is not a value and $C_1[M]$ is a value.
\caption{IPO contexts for the lazy/cbv lts's.}
\label{tb}
\end{center}
\end{figure}

The strong versions of context and IPO bisimilarities are too fine,
since they take into account reaction steps, and tell apart
$\beta$-convertible terms. Trivially, $I$ and $II$, where $I=\lambda
x.x$, are equivalent neither in the context bisimilarity nor in the
IPO bisimilarity, since $I\stackrel{[\ ]}{\not \rightarrow} $, while
$II\stackrel{[\ ]}{ \rightarrow} $ (both in the lazy and cbv case).
On the other hand, one can easily check that the weak context
bisimilarity, where the identity context $[\ ]$ is unobservable,
equates all closed terms.  The appropriate notion is that of weak IPO
bisimilarity, which, as we will see, turns out to capture exactly the
lazy and cbv equivalences.

It is interesting to observe that also the observational equivalence
and the applicative bisimilarity can be characterized as weak
bisimilarities on suitable context lts's.  In fact it is easy to prove
that the observational equivalence $\approx_{\sigma}$ coincides with
the weak bisimilarity on a restriction of the context lts built on
$\mathbf{C}_{\sigma}^{\lambda}$, defined by $M \stackrel{C[\
]}{\longrightarrow} N$ \iff\ $M \stackrel{C[\ ]}{\longrightarrow_C} N$
and $M\Downarrow_{\sigma}$. Similarly, the applicative equivalence can
be characterized by considering only applicative contexts in the lts.

In the following we will show that all these lts's induce the same notion of equivalence. 
Moreover, using the results of Section~\ref{ext}, we will show that the set 
of IPO contexts in the weak IPO bisimilarity to be considered can be significantly simplified. Then,
from
the fact that  the 
weak IPO lts is the smallest of the ones above, it follows that  it induces the simplest proofs that two terms are bisimilar.
 
Now, let us denote 
by
  $\approx_{\sigma I}$, for $\sigma \in \{ l,v\} $, the lazy/cbv \emph{weak IPO bisimilarity}, where the
  identity context is unobservable.
  In order to prove that $\approx_{\sigma I}$ is a congruence w.r.t. all contexts, we need to consider
  the category $\mathcal{D}_{\sigma}^{\lambda}$, list extension of $\mathcal{C}_{\sigma}^{\lambda}$, where the
  objects are finite lists $\langle 1, \ldots , 1 \rangle$, and an arrow 
\[\underbrace{\langle 1, \ldots , 1 \rangle}_n \rightarrow
\underbrace{\langle 1, \ldots , 1 \rangle}_m
\]
  is a m-tuple of possibly closed multi-holed contexts $\langle C_1,
  \ldots , C_m\rangle$ with n holes all together.  Multi-holed
  contexts are defined by
\[ C[\  ] \ ::= \ [\ ] \ |\ P \ | \ C[ \ ]C[\ ] \ |\ \lambda x.C[\ ] \ .\]
  Then, in the lazy case one can show that any closed multi-holed
  context either is IPO uniform or it is of the shape $[\ ] C_1[\ ]
  \ldots C_k[\ ]$ with the first hole reactive.  Namely, if $C[\ ]$ is
  of the shape $[\ ] C_1[\ ] \ldots C_k[\ ]$, then clearly the first
  hole is reactive.  Otherwise, it is of the shape $P C_1[\ ] \ldots
  C_k[\ ]$ or $(\lambda x. C_0[\ ] )C_1[\ ] \ldots C_k[\ ]$.  In the
  first case, the reduction (if any) involves only $P$ or at most
  $PC_1[\ ]$, where $C_1[\ ]$ together with the term put in the holes,
  plays only a passive role as argument. In the latter case, since the
  term put in the holes is closed, again it will be not affected by
  the substitution induced by the reduction.  Similarly, for the cbv
  case, all the multi-holed contexts are IPO uniform, apart from the
  contexts ranging on the following grammar, which have a reactive
  hole: \[ D[\ ]\ ::=\ [\ ]\ | \ D[\ ]C[\ ] \ |\ (\lambda x.C[\ ])D[\
  ] \ ,\] where $C$ is a closed multi-holed context. Moreover, the
  reduction relation is obviously deterministic.  Thus, by applying
  Proposition~\ref{list-weak-congruence}, we have:
 
 \begin{cor}\label{unocon} 
\hfill \begin{enumerate}[\em(i)]
\item
 For all $M, N \in \Lambda^0$, for any  closed unary context $C[\ ]$, \[ M \approx_{\sigma I} N \ \Longrightarrow \ 
 C[M]\approx_{\sigma I} C[N]\ .\]
 \item Moreover
 \[ \approx_{\sigma I}{}={}\approx_{\sigma R} \ ,\]
  where $\approx_{\sigma R}$ denotes the weak IPO bisimilarity where only reactive contexts are 
  considered (see the third columns in the tables of  Fig.~\ref{tb}).
\end{enumerate}
 \end{cor}
\noindent
Now, we are left to prove that the IPO bisimilarity coincides with the
original observational equivalence.  Notice that, in the above
proposition, we also provide a new alternative proof of the Context
Lemma for the lazy case.

\begin{prop}\label{sapp}
$\approx_{l}{}={}\approx^{app}_l{}={}\approx_{l I} $ and 
$\approx_{v}{}={}\approx_{v I} $.
\end{prop} 
\proof For the lazy case, we proceed by proving the following
chain of inclusions:
\begin{equation}\label{eqc} 
\approx_l{}\subseteq{}  \approx_{l}^{app}{}\subseteq{} \approx_{l R}{}
\subseteq{}\approx_{l I}{}\subseteq{}\approx_{l}\ .
\end{equation}
The first inclusion, $\approx_{l}{}\subseteq{} \approx_{l}^{app}$,
holds by definition.  The third inclusion, $\approx_{lR} {}\subseteq{}
\approx_{lI}$, follows by Corollary~\ref{unocon}(ii).  The others are
proved as follows:
\begin{enumerate}[$\bullet$]
\item $\approx_{l}^{app} {}\subseteq{} \approx_{l R}$. We prove that
$\approx_{l}^{app}$ is a ``weak IPO reactive bisimulation''. Let $M
\approx_{l}^{app} N$.  Assume $M \stackrel{C[\ ] }{\rightarrow_I} M'$
in the IPO reactive system.  By case analysis on $M$ and $C[\ ]$ we show
that $\exists N'.\ N\stackrel{C[\ ]}{\Rightarrow} N' \ \wedge\ M'
\approx_{l}^{app} N'$.

If $M \equiv (\lambda x. M_1)Q\vec{P}$ and $C[\ ] \equiv [\ ] $, then
$N \stackrel{[\ ]}{\Rightarrow}N$, $M'=_{\beta} M \approx_{l}^{app}
N$, hence closedness of $\approx_{l}^{app} $ under $\beta$-conversion
establishes the claim.

If $M \equiv \lambda x. M_1$ and $C[\ ] \equiv [\ ] P$, then, since
$M\approx_{l}^{app} N$, $\exists \lambda x. N_1 .\ N\stackrel{[\
]}{\Rightarrow} \lambda x.N_1 \stackrel{[\ ]P}{\rightarrow} N'$.  Then
$M'=_{\beta} MP \approx_{l}^{app} NP =_{\beta} N'$, and closedness of
$\approx_l$ under $\beta$-conversion establishes the claim.

\item $\approx_{l I}{}\subseteq{} \approx_{l}$. Let $M\approx_{l I}
N$. We have to show that, for any unary closed context $C[\ ]$, $C[M]
\Downarrow \ \Leftrightarrow \ C[N]\Downarrow$.  From $M\approx_{l I}
N$, by Corollary~\ref{unocon}(i), we have $C[M]\approx_{l I} C[N]$.
Now assume that $C[M] \Downarrow_l$, then there exists $M'$ such that
$C[M] \stackrel{[\ ]P}{\Rightarrow} M'$, hence also there exists $N'$
such that $C[N] \stackrel{[\ ]P}{\Rightarrow} N'$, thus $C[N]
\Downarrow_l$.
\end{enumerate}
The above argument  provides a new proof of the Context Lemma. 

For the cbv case, considering the applicative equivalence $\approx_v^{app}$ does not help, but
one can prove directly:
\begin{equation}  
\approx_{v} {}\subseteq{} \approx_{v R}{}\subseteq{} \approx_{v I}
{}\subseteq{} \approx_{v} \label{equno}
\end{equation}
\begin{enumerate}[$\bullet$]
\item $\approx_{v} {}\subseteq{} \approx_{v R}$.  One can easily check that $\approx_{v}$ is
a ``weak IPO reactive bisimulation'', using the fact that $\approx_v$ is closed under $\beta$-reduction.
\item $\approx_{vR} {}\subseteq{} \approx_{v I}$. Immediate by Corollary~\ref{unocon}(ii).
\item  $\approx_{vI} {}\subseteq{} \approx_{v }$. Let $M\approx_{vI} N$. We have to show that, for any 
unary context $C[\ ]$, $C[M]\Downarrow_v \Longleftrightarrow C[N]\Downarrow_v$. From 
$M\approx_{vI} N$, by Corollary~\ref{unocon}(i), we have $C[M] \approx_{vI}  C[N]$.  
Now assume that $C[M] \Downarrow_v$, then there exists $M'$ such that $C[M] \stackrel{[\ ]V}{\Rightarrow} M'$, hence also
there exists $N'$ such that $C[N] \stackrel{[\ ]V}{\Rightarrow} N'$, thus $C[N] \Downarrow_v$.\qed
\end{enumerate}

 \begin{rem}
 Corollary~\ref{unocon}(ii) allows us to reduce the set of IPO contexts to be considered in the 
 IPO bisimilarities. For the lazy case, only applicative contexts can be considered (see the first table in Figure~\ref{tb}), while for the cbv case, the set of reactive IPO contexts is larger (see the second 
 table in Figure~\ref{tb}). However, also for the cbv case, one can prove that applicative  (by-value)
 IPO contexts are sufficient. We omit the details.
 \end{rem}

Proposition~\ref{sapp} above gives us  interesting characterizations of  lazy and cbv observational 
equivalences, in terms of lts's where the labels are significantly reduced. 
However, such lts's (and bisimilarities) are still infinitely branching, e.g.
$\lambda x.M \itrans{P}$, for all $P\in \Lambda^0$.
This is due to the fact that  the  context  categories underlying 
the reactive systems $\mathbf{C}^{\lambda}_l$ and $\mathbf{C}^{\lambda}_v$ allow only for a  ground representation
of the $\beta$-rule through infinitely many ground rules.
 In order
to overcome  this problem, one should look for alternative  categories which allow  for a parametric representation of the $\beta$-rule as $(\lambda x. X )Y \rightarrow X[Y/x]$, where $X,Y$
are parameters. To this aim, we introduce the category of \emph{second-order term contexts}
(see  Section~\ref{sec} below).  However, as we will see, 
this approach works only if the reaction rules are ``local'', that is, they do not 
act on the whole term, but only locally. In particular, the operation of substitution on the 
$\lambda$-calculus is not  local and thus it is not describable by a finite set of reaction rules.
To avoid this problem, in the following section we consider  encodings of the $\lambda$-calculus into  Combinatory Logic  (CL) endowed with
suitable strategies and equivalences, which turn out to correspond to lazy and cbv
equivalences. 

\section{Combinatory Logic} \label{CL}
In this section, we focus on \emph{Combinatory Logic}  \cite{HS86} with Curry's combinators  $\K,\SSS$, and we study its relationships
with the $\lambda$-calculus endowed with lazy and cbv reduction strategies. 
An interesting result that we prove is that we can define suitable reduction strategies on CL-terms, inducing observational equivalences which correspond to lazy and cbv equivalences on 
$\lambda$-calculus.
 As a consequence, we can safely shift our attention from the reactive system of $\lambda$-calculus 
 to the simpler reactive system of CL.  In this section, we apply Leifer-Milner construction to CL 
 viewed as a (standard) context category, and we study weak versions of context and  IPO bisimilarities. Our main
 result is that we can recover lazy and cbv observational equivalences as weak IPO equivalences on 
 CL$^*$, a variant of standard CL. Here the approach is first-order, thus the IPO equivalences are still
 infinitely branching. However, the results in this section are both interesting in themselves, and 
 useful for our subsequent investigation of Section~\ref{sec}, where CL is viewed as a second-order
 rewriting system, and a characterization of the lazy observational equivalence as a finitely branching IPO
 bisimilarity is given.

In \cite{Sew02}, a construction, similar to Leifer-Milner construction, has been applied to the Combinary Logic case.  However, in that paper, it has been left open the question of whether the weak bisimilarity on the derived LTS is a congruence.  In this paper, using Proposition \ref{list-weak-congruence}, we can positively answer that question. 

\begin{defi}[Combinatory Terms]\label{comt} The set of combinatory terms is defined by:
\[ (CL  \ni)\  M\ ::=\   x\ | \ \K\  |\  \SSS\  | \ MM \ , \]
where $\K$, $\SSS$ are combinators.  Let $CL^0$ denote the set of
\emph{closed} CL-terms.
\end{defi}

\subsection{Correspondence with the \texorpdfstring{$\lambda$}{Lambda}-calculus}
 
Let $\Lambda (\K,\SSS)$ denote the set of $\lambda$-terms built over constants $\K,\SSS$. The following is a well-known encoding:

\begin{defi}[$\lambda$-encoding]
Let $\mathcal{T} : \Lambda (\K,\SSS)  \rightarrow CL$  be the transformation defined as follows:
\[\eqalign{
  \mathcal{T} (x)
&=x\cr
  \mathcal{T}(MN)
&=\mathcal{T}(M) \mathcal{T}(N)\cr
  \mathcal{T}(\lambda x.x)
&=\SSS \K\K\cr
  \mathcal{T}(\lambda x.y)
&=\K y\cr
  }
  \qquad\qquad
  \eqalign{
  \mathcal{T}(C)
&=C\enspace\mbox{if}\enspace C\in \{ \K, \SSS \}\cr 
  \mathcal{T} (\lambda x. MN)
&=\SSS \mathcal{T}(\lambda x.M)\mathcal{T}(\lambda x.N)\cr
  \mathcal{T}(\lambda x.\lambda y.M)
&=\mathcal{T} (\lambda x. \mathcal{T} (\lambda y. M) )\cr
  \mathcal{T} (\lambda x.C)
&=\K \mathcal{T} (C)\enspace\mbox{if}\enspace C\in \{ \K, \SSS \}\cr
  }
\]

\noindent
In particular, if we restrict the domain of $\mathcal{T}$ to
$\Lambda$, we get an encoding of $\lambda$-terms into CL.  Vice versa,
there is a natural embedding of CL into the $\lambda$-calculus
$\mathcal{E}: CL \rightarrow \Lambda$: 
\[
\mathcal{E} (\K) =\lambda xy.x\quad
\mathcal{E}(\SSS) = \lambda xyz. (xz)(yz)\quad
\mathcal{E} (x) =x\quad
\mathcal{E}(MN) = \mathcal{E} (M) \mathcal{E}(N)
\]
\end{defi}

The following  lemma holds: 

\begin{lem} \label{es}
For all $M\in \Lambda$,\  $\mathcal{E} (\mathcal{T}(M))=_{\sigma} M$, for $\sigma \in \{ \beta, \beta_V\}$. 
\end{lem}
\proof First, one can easily prove that, if $M$ is $\lambda$-free,
 then $\mathcal{E} \mathcal{T} (\lambda x.M) =_{\sigma}\lambda x.M$
 (by induction on $M$). Then, using the fact that $\mathcal{T} (M)$ is
 $\lambda$-free for all $M$, by definition of $\mathcal{T}$, one gets
 that $\mathcal{T}^2 (M)= \mathcal{T} (M)$ for all $M$. Finally, we
 are ready to prove the claim in its full generality by induction on
 $ M$. The only non-trivial case is when $M= \lambda x. \lambda y.N$.
 Then we have $ \mathcal{E} \mathcal{T}( \lambda x. \lambda y.N)=
 \mathcal{E} \mathcal{T}( \lambda x. \mathcal{T}( \lambda y.N))$,
 where $\mathcal{T}( \lambda y.N)= PQ$ is $\lambda$-free.  Then
\[\eqalign{
  \mathcal{E} \mathcal{T}( \lambda x. \lambda y.N)
&=\mathcal{E} \mathcal{T}(\lambda x. PQ)
 =\mathcal{E}({\mathbf S})\mathcal{E}\mathcal{T}(\lambda x.P)
                          \mathcal{E}\mathcal{T}(\lambda x.Q)\cr
&=(\lambda xyz. (xz)(yz)) \mathcal{E}\mathcal{T}(\lambda x.P)  
                          \mathcal{E}\mathcal{T}(\lambda x.Q)\cr
&=_{\sigma} ( \lambda xyz. (xz)(yz)) \lambda x. P \lambda x. Q\,,\enspace 
  \mbox{since $PQ$ is $\lambda$-free,}\cr 
&=_{\sigma} \lambda x.PQ 
=_{\sigma} \lambda x. \mathcal{E}\mathcal{T} (PQ)\,,\enspace
  \mbox{since $PQ$ is $\lambda$-free,}\cr 
&=\lambda x. \mathcal{E}\mathcal{T}\mathcal{T} (\lambda y.N)
 =\lambda x. \mathcal{E}\mathcal{T}(\lambda y.N)\,,\enspace
  \mbox{since $\mathcal{T}^2 = \mathcal{T}$,}\cr 
&=_{\sigma} \lambda x. \lambda y.N\,,\enspace
  \mbox{by induction hypothesis.}\hbox to140 pt{\hfil\qEd}\cr
  } 
\]

\subsubsection{Lazy/cbv observational equivalence on CL}
Usually, the set of combinatory terms are endowed with the following reaction
rules: 
\[ \K MN\rightarrow M \ \  \ \ \ \  \ \  \ \ \ \SSS MNP \rightarrow (MP)(NP) 
\]  
  We will also consider a cbv version of the above rules, reducing CL
  redexes only when the arguments are values, i.e., terms on the
  following grammar:
\[   V\ :: = \  \K \ |\ \SSS\ | \  \K V \ | \SSS V\ |\ 
 \SSS VV \ .\]
  The cbv rules are the following:
\[\K V_1V_2\rightarrow V_1\qquad\qquad\SSS V_1V_2V_3\rightarrow(V_1V_3)(V_2 V_3) 
\]  

\begin{defi}[Lazy/cbv Reduction Strategy on CL] \label{lazyr}
\hfill \begin{enumerate}[(i)]
\item The \emph{lazy reduction strategy} $\rightarrow_l{}\subseteq{} CL^0 \times CL^0$  reduces the  leftmost outermost  CL-redex.
Formally:
\[\reguno{}{\SSS M_1 M_2 M_3 \rightarrow_l (M_1 M_3) (M_2 M_3)} \qquad
  \reguno{}{\K M_1 M_2  \rightarrow_l M_1 } \qquad
  \reguno{M \rightarrow_l M' }{MP  \rightarrow_l M'P }
\]\smallskip
\item The \emph{cbv strategy} $\rightarrow_v{}\subseteq{} CL^0 \times CL^0$  
is defined by
\[\reguno{}{\SSS V_1 V_2 V_3  \rightarrow_v (V_1 V_3) (V_2 V_3)}\qquad
  \reguno{}{\K V_1 V_2  \rightarrow_v V_1 } \qquad 
  \reguno{M_1 \rightarrow_v M'_1}{\K M_1  \rightarrow_v \K M'_1}
\]
\[\reguno{M_2 \rightarrow_v M'_2}{\K V_1 M_2 \rightarrow_v \K V_1 M'_2}\qquad
  \reguno{M_1 \rightarrow_v M'_1}{\SSS M_1 \rightarrow_v \SSS M'_1}\qquad
  \reguno{M_2 \rightarrow_v M'_2}{\SSS V_1 M_2  \rightarrow_v \SSS
  V_1M'_2}
\]
\[\reguno{M_3 \rightarrow_v M'_3}
  {\SSS V_1 V_2 M_3 \rightarrow_v \SSS V_1 V_2 M'_3} \qquad
  \reguno{M \rightarrow_v M'}{MP \rightarrow_v M'P}
\]\medskip

\noindent where $V_1, V_2, V_3$ are values.
\end{enumerate} 
\end{defi}
 
 \begin{defi}[Unary Contexts on CL] \label{uccl}
 The set of \emph{unary contexts} on CL is defined by
 \[ C[\ ]\ ::=\  [\ ] \ | \ C[\ ]P \  |\ PC[\ ] \ . \]
 \end{defi}
 
Alternatively we could define the lazy strategy $\rightarrow_l$ as the closure of the standard 
 CL-reaction rules under the following reactive
 contexts (which coincide with the applicative ones):
 \[ D[ \ ] \ ::= \ [\ ] \ | \ D[\ ]P \ . \]
 Similarly, 
 we could define the cbv strategy  $\rightarrow_v$ as the closure of the cbv reaction rules under the following reactive
 contexts: 
 \[ D[ \ ] \ ::= \  [\ ] \ | \   D[\ ]P \ | \ \K D[\ ] \ |\  \K V D[\ ]  \ | \  \SSS  D[\ ] \ | \ \SSS V D[\ ] \ | \ \SSS V_1 V_2 D[\ ].  
\]
Let $\downarrow_{\sigma}$  denote  the  convergence relation on CL, for $\sigma \in \{ l,v\}$.
 
\begin{defi}[Lazy/cbv Equivalence on CL] \label{dleq}
\hfill \begin{enumerate}[(i)]
\item A  relation  $\mathcal{R}\subseteq CL^0 \times CL^0$ is a
\begin{enumerate}[$-$]
\item  \emph{CL lazy bisimulation} if:
\[\langle M,N \rangle \in {\mathcal R} \ \Longrightarrow \ (M\downarrow_{l} \ \Leftrightarrow \ N\downarrow_{l}) 
\ \wedge \
\forall P\in CL^0.\  \langle MP, NP\rangle\in
{\mathcal R}\ .
\]
\item  \emph{CL cbv bisimulation} if:
\[\langle M,N \rangle \in {\mathcal R} \ \Longrightarrow \ (M\downarrow_{v} \ \Leftrightarrow \ N\downarrow_{v}) 
\ \wedge \
\forall \mbox{ closed value }V\in CL^0.\  \langle MV, NV\rangle\in
{\mathcal R}\ .
\]
\end{enumerate}
\item Let $\simeq_{\sigma}{}\subseteq{} CL^0 \times CL^0$ be the largest CL lazy/cbv bisimulation.
\item Let $\widehat{\simeq}_{\sigma}{}\subseteq{} CL \times CL$ denote the extension of  $\simeq_{\sigma}$
to open terms defined by:  for $M,N\in CL$ s.t.\ $FV(M,N)\subseteq \{ x_1, \ldots , x_n\}$, $ M \widehat{\simeq}_{\sigma} N $ \iff\ for all closing (by-value) substitutions $\vec{P}$, $M[\vec{P}/\vec{x}] \simeq_{\sigma}
N[\vec{P}/\vec{x}]$.
\end{enumerate}
\end{defi}

\noindent
Notice that we use two different symbols for equivalences ($\approx$ and $\simeq$), in this way we distinguish the equivalence relation on $\lambda$-terms from the corresponding relation on CL. 

The following theorem  is interesting \emph{per se}:

\begin{thm}\label{eqv}
For all $M,N \in \Lambda$, $M\,\widehat{\approx}_{\sigma}\, N
\ \Longleftrightarrow \ \mathcal{T} (M)\,\widehat{\simeq}_{\sigma}\,
\mathcal{T}(N) \ . $
\end{thm} 

\noindent {\bf Proof of Theorem~\ref{eqv}}.
We carry out the proof of the above theorem for the lazy case, the proof for the cbv case being similar.

\begin{lem}\label{A}
\hfill \begin{enumerate}[\em(i)]
\item For all $M\in CL^0$, $M\dlcl \ \Longleftrightarrow \ \mathcal{E} (M)\dll$.
\item
For all  $M\in \Lambda^0$,
$M\dll \ \Longleftrightarrow \ \mathcal{T}(M) \dlcl .$
\end{enumerate}
\end{lem}

\proof\hfill
\begin{enumerate}[(i)]
\item  By definition of the lazy strategies on $\lambda$-terms and on CL-terms.
\item ($\Rightarrow$) Let $M\dll$. Then, since by Lemma~\ref{es} $\mathcal{ E}(\mathcal{ T}(M))=_{\beta} M$, and $\eqll$ is closed under $\beta$-conversion, we have also
$\mathcal{ E}(\mathcal{ T}(M)) \dll$. Thus, by (i), $\mathcal{ T} (M) \dlcl$.
\item[] ($\Leftarrow$) Let $\mathcal{ T} (M) \dlcl$. By (i), 
$\mathcal{ E}(\mathcal{ T} (M)) \dll$, by Lemma~\ref{es}, $M =_{\beta} \mathcal{ E}(\mathcal{ T} (M))$,  thus $M\dll$.\qed
\end{enumerate}

\begin{lem} \label{B}
For all $M,N \in CL^0$, if $\mathcal{ E} (M) =_{\beta} \mathcal{ E} (N)$, then $M\eqlcl  N$.
\end{lem}
\begin{proof} The proof follows from the fact that ${\mathcal R}=\{ \langle M,N \rangle \in CL^0 \mid \mathcal{ E} (M) 
=_{\beta}  \mathcal{ E} (N) \}$ is a CL lazy bisimulation. Namely $M\dlcl $ \iff\ $N\dlcl$,
because, by Lemma~\ref{A}(i), 
 $M\dlcl$ \iff\ $\mathcal{ E} (M) \dlcl$ and $N\dlcl$ \iff\ $\mathcal{ E} (N) \dlcl$, and $\approx^{app}_l$ is closed under $\beta$-conversion. Moreover, for any $P\in
CL^0$,
$\langle MP, NP\rangle \in {\mathcal R}$, since $\mathcal{ E} (MP) =\mathcal{ E}(M) \mathcal{ E} (P) =_{\beta}
\mathcal{ E}(N) \mathcal{ E} (P) =\mathcal{ E} (NP)$.
 \end{proof}

\begin{lem}\label{AA}  
$\forall P\in CL^0$, $P\simeq_l \mathcal{ T} (\mathcal{ E}(P)).$
\end{lem}
\begin{proof} We prove that 
$\mathcal{ R} =\{ (P \vec{R}, \mathcal{ T} (\mathcal{ E}(P))\vec{R}) \mid P, \vec{R} \in CL^0 \}$ is a bisimulation.
To this aim, it is sufficient to prove that, for all $P, \vec{R}$, $P\vec{R} \downarrow_l \ \Leftrightarrow\
 \mathcal{ T} (\mathcal{ E}(P))\vec{R} \Downarrow_l$. By Lemma~\ref{A}, 
$P\vec{R} \downarrow_l \ \Leftrightarrow\
 \mathcal{ E}(P\vec{R}) \Downarrow_l$. 
 Now $\mathcal{ E}(P\vec{R}) = \mathcal{ E}(P)  \mathcal{ E}(\vec{R}) =_{\beta} ( \mathcal{ E} \circ \mathcal{ T} \circ
 \mathcal{ E} (P)) \mathcal{ E} (\vec{R})  = \mathcal{ E} ((\mathcal{ T}\circ \mathcal{ E}(P)) \vec{R})$. 
 \\ Finally,
 by Lemma~\ref{A}, $\mathcal{ E} ((\mathcal{ T}\circ \mathcal{ E}(P)) \vec{R} \Downarrow_l \ \Longleftrightarrow\
 \mathcal{ T}(\mathcal{ E}(P)) \vec{R} \downarrow_l$. 
\end{proof}

\begin{lem}\label{AAA}
Let $M\in \Lambda$ and let $\vec{P}$ be closed such that $M[\vec{P}/\vec{x}] \in \Lambda^0$, then \[ \mathcal{T} (M[\vec{P}/\vec{x}]) 
\simeq_{l} \mathcal{ T} (M) [\mathcal{T}(\vec{P})/\vec{x}]\ .\] 
\end{lem}
\begin{proof}
By Lemma~\ref{B}, it is sufficient to show that ${\mathcal E} (\mathcal{T} (M[\vec{P}/\vec{x}]) ) =_{\beta}
{\mathcal E} (\mathcal{ T} (M) [\mathcal{T}(\vec{P})/\vec{x}])$. 
Now 
${\mathcal E} (\mathcal{T} (M[\vec{P}/\vec{x}]) )=_{\beta} M[\vec{P}/\vec{x}]$, by Lemma~\ref{es}. 
On the other hand, from the definition of ${\mathcal E}$, one can prove by induction that 
${\mathcal E} (\mathcal{ T} (M) [\mathcal{T}(\vec{P})/\vec{x}])=  {\mathcal E} \mathcal{T} (M) [{\mathcal E} {\mathcal T} (\vec{P})/\vec{x}]$,
which, by Lemma~\ref{es}, 
$=_{\beta} M[\vec{P}/\vec{x}]$.
\end{proof}

Now we proceed to prove Theorem~\ref{eqv} $(\Rightarrow)$. Assuming $M\widehat{\approx}_{l}  N$, we have to prove that,
for all closing $\vec{P}$, ${\mathcal T} (M) [\vec{P}/\vec{x}] \simeq_l {\mathcal T} (N) [\vec{P}/\vec{x}] $.  From
$M\widehat{\approx}_{l}  N$ it follows  $M[{\mathcal E}(\vec{P})/\vec{x}] \approx_l N[{\mathcal E}(\vec{P})/\vec{x}]$.
By Lemmata~\ref{A},~\ref{AA}, using the fact that $\simeq_l$ is a congruence, we have
${\mathcal T} (M[{\mathcal E}(\vec{P})/\vec{x}]) \simeq_l {\mathcal T} (N[{\mathcal E}(\vec{P})/\vec{x}])$. By Lemma~\ref{AAA}, 
${\mathcal T} (M)[{\mathcal T} {\mathcal E}(\vec{P})/\vec{x}] \simeq_l {\mathcal T} (N)[{\mathcal T} {\mathcal E}(\vec{P})/\vec{x}]$, hence
by Lemma~\ref{AAA}, using the fact that $\simeq_l$ is a congruence, we have 
${\mathcal T} (M)[\vec{P}/\vec{x}] \simeq_l {\mathcal T} (N)[\vec{P}/\vec{x}]$.

In order to prove Theorem~\ref{eqv} $(\Leftarrow)$, assume ${\mathcal T}(M) \widehat{\simeq}_{l}  {\mathcal T}(N)$. We have to
prove that, for all closing $\vec{P}$, $M[\vec{P}/\vec{x}]\approx_l N[\vec{P}/\vec{x}]$. From 
${\mathcal T}(M) \widehat{\simeq}_{l}  {\mathcal T}(N)$ it follows 
${\mathcal T}(M)[{\mathcal T} (\vec{P})/\vec{x}]  \simeq_{l}  {\mathcal T}(N)[{\mathcal T} (\vec{P})/\vec{x}] $.
From Lemma~\ref{AAA}, we have 
${\mathcal T}(M[\vec{P}/\vec{x}] ) \simeq_{l}  {\mathcal T}(N [\vec{P}/\vec{x}] )$. By Lemma~\ref{A}, we have
$M[\vec{P}/\vec{x}] \approx_l N[\vec{P}/\vec{x}]$.

\subsection{The First-order Approach: CL as a Context Category}\label{first}
We endow CL with a structure of reactive system in the sense of \cite{LM00}, by considering the context category
 of closed unary contexts:  

\begin{defi}[Lazy, cbv CL Reactive Systems] \label{lazyRS}
${\mathbf C}^1_{\sigma}$, for $\sigma \in \{ l,v\}$,  consists of:
\begin{enumerate}[$\bullet$]
\item the context category whose objects are $0,1$, where the morphisms from 
0 to 1  are the closed terms, the morphisms from 1 to 1 are the  
closed unary contexts, and composition is context substitution;
\item the subcategory of reactive contexts is determined by the reactive contexts for the lazy and cbv strategy, respectively,  presented in Definition~\ref{lazyr}; 
\item the reaction rules are the standard CL reduction rules for the lazy case, and the cbv
reduction rules for the cbv case.
\end{enumerate}
\end{defi}

\begin{lem}
The reactive systems ${\mathbf C}^1_{\sigma}$ have redex RPOs.
\end{lem} 

One can easily check that the IPO contexts are  the following.
\begin{enumerate}[$\bullet$]
\item {\bf Lazy}.  The IPO contexts for a given term $M$ are:
\begin{enumerate}[$-$] 
\item  $[\ ]\vec{P}$, where
$\vec{P}$ has the minimal length for the top-level reaction of $M$ to fire,
\item   $\K C[\ ]P_1,\
  \K P_1 C[\ ], \ \K P_1 \vec{Q} C[\ ]$, for any  $ C[\ ], \vec{Q}, P_1$,
\item $\SSS C[\ ]P_1 P_2,\
\SSS P_1 C[\ ]P_2,\   \SSS P_1 P_2 C[\ ],\   \SSS P_1 P_2 \vec{Q} C[\ ]$, for any  $ P_1, P_2,C[\ ], \vec{Q}$.
\end{enumerate}
\item {\bf Cbv}.
\\ For $M$  not a value, the following contexts are IPOs:
\begin{enumerate}[$-$]
\item   $[\ ]$, 
\end{enumerate}
\noindent  For $M$ value, the following contexts are IPOs:
\begin{enumerate}[$-$]
\item   $[\ ] V_1 \ldots V_i$,
where $i$ is the minimum number of arguments necessary for the top-level reaction of $M$ to fire,   
\item $[\ ]  V_1 \ldots V_i P$, where $P$ is not a value, and $i$, possibly $0$, is less than the minimum number of arguments necessary for the top-level reaction of $M$ to fire,
\item $V C[\ ] V_1 \ldots V_i$
where $V$ and $C[M]$ are values and $i+1$ is the minimum number of arguments necessary for the top-level reaction of $V$ to fire,  in more detail:  $\K C[\ ] V$, $\K V C[\ ]$, $\SSS C[\ ] V_1 V_2$,  $\SSS V_1 C[\ ] V_2$, $\SSS V_1 V_2 C[\ ]$, 
\item $V C[\ ]  V_1 \ldots V_i P$, where $V$ and $C[M]$ are values, $P$ is not a value, and $i+1$ is less than the minimum number of arguments necessary for the top-level reaction of $V$ to fire, in more detail: $\K C[\ ] P$, $\SSS C[\ ] P$, $\SSS C[\ ] V_1 P$,  $\SSS V_1 C[\ ] P$.
\end{enumerate}
For any term $M$, the following contexts are IPOs: 
\begin{enumerate}[$-$]
\item $P C[\ ]$, where $P$ is not a value and $C[\ ]$ is any context.
\end{enumerate}
\end{enumerate}
\noindent

For any of the above contexts there is a reduction rule which applies, and the context is minimal for the given reduction to fire.
By case analysis, one can show that all the other contexts are not IPO contexts.

The strong versions of context and IPO bisimilarities are too fine, since, as in the $\lambda$-calculus
case, they take into account reduction steps, and
 tell apart $\beta$-convertible terms. Thus we consider weak variants of such equivalences, where
 the identity context $[ \ ]$ is unobservable.
 Weak context bisimilarity is too coarse, since  it equates all terms.
 However,  we will prove that the weak IPO  bisimilarity ``almost'' coincides with the lazy/cbv equivalence. Moreover, we will show how to recover the exact correspondence by considering a suitable variant of
 CL. 
 
First of all, let $\simeq_{\sigma I}$, for $\sigma \in \{ l,v\}$,  denote the \emph{lazy/cbv weak IPO bisimilarity} obtained 
by considering the identity context as unobservable.
Similarly to the case of the $\lambda$-calculus, we can define a list extension category by taking the category of multi-holed contexts.  In this category  all contexts with no reactive indexes are IPO uniform.
In the lazy case, the contexts with a reactive index are of the shape $[\ ]C_1[\ ]\ldots C_k[\ ]$
(with the leftmost hole being reactive), and the remaining ones have not reactive indexes and are IPO uniform.
 For the cbv case,  one can show that the multi-holed contexts with a reactive index are given by the grammar:
 \[ D[ \ ] \ ::= \  [\ ] \ | \   D[\ ]C[\ ] \ | \ \K D[\ ] \ |\  \K V D[\ ]  \ | \  \SSS  D[\ ] \ | \ \SSS V D[\ ] \ | \ \SSS V_1 V_2 D[\ ] \ ,
\]
where $C[\ ]$ is any closed multi-holed context.

Thus, by Proposition~\ref{list-weak-congruence}(i),  we have: 
 
 \begin{prop}\label{cunocon}
 For all $M, N \in CL^0$,  for any closed unary context $C[\ ]$,
 \[ M \simeq_{\sigma I} N \ \Longrightarrow \ C[M] \simeq_{\sigma I} C[N] \ .\]
 \end{prop}

The rest of this section is devoted to compare the lazy/cbv weak IPO bisimilarity $\simeq_{\sigma I}$ with
 the lazy/cbv equivalence on CL $\simeq_{\sigma}$ defined in Definition~\ref{dleq}.  The following lemma
 can be easily proved by coinduction, using  Proposition~\ref{cunocon}.
 
\begin{lem}\label{ccont}
$\simeq_{\sigma I}{}\subseteq{} \simeq_{\sigma}$.
\end{lem}
\begin{proof}
We prove that $\simeq_{\sigma I} $ is a lazy/cbv bisimulation 
on CL.  Let $M\simeq_{\sigma I} N $. If $M\downarrow_{\sigma}$, then also $N\downarrow_{\sigma}$, since a convergent term has 
different IPO-transitions from
a divergent term. We are left to prove that for all $P$, $MP \simeq_{\sigma I} N P$. But 
this follows from Proposition~\ref{cunocon}.
\end{proof}

However, the converse inclusion  $ \simeq_{\sigma}  {}\subseteq{} \simeq_{\sigma I}$ does not hold, 
since 
for instance  $\K \simeq_{\sigma}  \SSS (\K \K) (\SSS \K \K)$, because, e.g. for the lazy case, for all $P$,
 $ \SSS (\K \K) (\SSS \K \K)P \rightarrow^* \K P$. But 
 $\K\not  \simeq_{\sigma I} \SSS (\K \K) (\SSS \K \K)$. Namely
$ \SSS (\K   \K)(\SSS \K \K) \itrans{[\ ] V}$, while $ \K  \stackrel{[\ ] V}{\not\rightarrow_I}$.
The problem, which was already noticed in \cite{Sew02}, 
arises since the equivalence $ \simeq_{\sigma I}$ tells apart
terms whose top-level combinators
expect a different number of arguments to reduce.
In order to overcome this problem,  we consider an extended calculus, CL$^*$, where the combinators 
$\K$ and $\SSS$ become unary, at the price of adding new intermediate combinators and 
intermediate reductions (the reactive contexts are the ones in Definition~\ref{lazyRS}). 

\begin{defi}
The CL$^*$ lazy combinatory calculus is defined by
\begin{enumerate}[$\bullet$]
\item Terms: 
\[M\ ::=\   x\ | \ \K\  |\  \SSS\  |\  \K' M \   |\   \SSS' M \ |\
\SSS'' MN \ |  \ MN 
\] 
  where $\K$, $\K'$, $\SSS$, $\SSS'$,  $\SSS''$ are combinators.
\item Rules:
\[\K M\rightarrow \K' M \quad \K'MN \rightarrow M\]
\[\SSS M \rightarrow \SSS'M \quad \SSS'M N \rightarrow \SSS''MN \quad
         \SSS'' MNP \rightarrow (MP)(NP) 
\]
\end{enumerate}

\noindent The CL$^*$ cbv combinatory calculus is defined by
\begin{enumerate}[$\bullet$]
\item Terms: 
\[M\ ::=\   x\ | \ \K \ |\  \SSS  \ | \ MN  \ |\  \K' V \ |\   \SSS' V
\ |\  \SSS'' V V
\] 
  Values: 
\[V\ ::= \K \ | \ \K' V \ |\  \SSS  \  | \ \SSS' V \  | \ \SSS'' V V \]
  where $\K$, $\K'$, $\SSS$, $\SSS'$,  $\SSS''$ are combinators.
\item Rules:
\[\K V_1\rightarrow \K' V_1 \quad \K'V_1 V_2 \rightarrow V_1 \]
\[\SSS V_1 \rightarrow \SSS'V_1 \quad \SSS' V_1 V_2 \rightarrow \SSS''V_1 V_2
    \quad\SSS'' V_1 V_2 V_3  \rightarrow (V_1V_3)(V_2 V_3) 
\]
\end{enumerate}
\end{defi}

\noindent
Notice that the calculus in the above definition is well-defined, since the set of terms is closed under
the reaction rules. One can define lazy/cbv reduction strategies on CL$^*$  as
in Definition~\ref{lazyr}, or as the closures of the reaction rules under the following reactive contexts:

\begin{defi}[CL$^*$ Reactive Contexts]\hfill
\begin{enumerate}[$\bullet$]
\item {\bf Lazy}. \ $ D[\ ] \ ::= \ [\ ] \ | \ D[\ ] P $ .
\item {\bf Cbv}. \ $D[\ ] \ ::= \ [\ ] \ | \ D[\ ]P \ |\ V D[\ ]$.
\end{enumerate}
\end{defi}

\noindent
Let $\simeq^*_{\sigma}$ be the lazy/cbv equivalence defined on CL$^*$,
similarly as in Definition~\ref{dleq} for CL.  There is a trivial
embedding of CL-terms into CL$^*$. Moreover, one can easily check
that, when restricted to terms of CL, $\simeq^*_{\sigma}$ coincides
with $\simeq_{\sigma}$.

Analogously to the CL case, we define the reactive system over CL$^*$.  In the context category, the unary closed contexts are defined by the grammar 
\[ C[\ ] \ ::= \ [\ ] \ |\ C[\ ] M   \ |\ M C[\ ] \]
where $M$ is a closed term.  Notice that, under the above definition, expressions like $\K' [\ ]$ do not represent unary closed context. In defining the IPO transitions, it is important to observe that $C[M]$ is a value \iff\ $M$ is a value and $C[\ ]$ is the identity context $[\ ]$. 
Let us denote by $\simeq_{\sigma I}^{*}$ the weak IPO bisimilarity obtained by considering the lazy/cbv  reactive system over CL$^*$. Since CL$^*$-terms expect at most one argument,
the IPO contexts for CL$^*$ are simpler than the ones for CL, and they are summarized in 
Figure~\ref{tbstar}.

\begin{figure}
\begin{center}
\begin{tabular}{| l |l | l |}
\multicolumn{3}{c}{\bf Lazy IPO  lts's on  CL$^*$}\\
\hline
term $M$ & IPO contexts & reactive IPO contexts \\
\hline\hline
 $M$ value &  $[\ ]P,$ $P C[ \ ]$ & $[\ ]P$ \\
\hline 
$M$ not a value & $[\ ] $,  $P C[ \ ]$  & $ [\ ] $\\
\hline
\end{tabular} \vspace*{0.5cm}

\begin{tabular}{| l |l | l| }
\multicolumn{3}{c}{\bf Cbv IPO lts's on  CL$^*$}\\
\hline
term  $M$ & IPO contexts & reactive IPO contexts \\
\hline\hline  \vspace*{0.5ex}
 $M$ value  &  $[\ ]P,$ $RC[ \ ]$, $V [\ ]$  & $[\ ]P$, $V [\ ]$ \\
\hline
$M$ not a value & $[\ ] $,  $RC[\ ]$ & $[\ ]$  \\
\hline
\end{tabular} 

\bigskip
 where  $R$ is not a value, $V$ is a value, $C[\ ]$ is a generic unary context.
\caption{IPO contexts for the lazy/cbv lts's on  CL$^*$.}
\label{tbstar}
\end{center}
\end{figure}

Similarly to the previous case, one can consider the multi-holed
contexts category as a list extension category.  In this category all
contexts are either IPO uniform or have a reactive index.  Moreover,
the reduction relation is deterministic.  Thus
Proposition~\ref{list-weak-congruence} applies and we have:

\begin{prop}\label{costar}
\hfill \begin{enumerate}[\em(i)]
\item
The equivalence $\simeq_{\sigma I}^{*}$ is a congruence w.r.t. unary contexts.
\item $\simeq^*_{\sigma I}{}={}\simeq^*_{\sigma R}$, where $\simeq^*_{\sigma R}$ denotes the IPO bisimilarity where only reactive IPO contexts are considered.
\end{enumerate}
\end{prop}

\noindent
By Proposition~\ref{costar}(ii) above, 
 the weak IPO equivalence can be significantly 
simplified.
Namely,  in the lazy case, we obtain   the weak IPO bisimilarity $\simeq_{l R}$, where
 only applicative  IPO contexts are considered (see Figure~\ref{tbstar}).
In the cbv case, Proposition~\ref{costar} 
allows us to
 reduce ourselves  to contexts of the shape $[\ ], [\ ] P, V [\ ]$ (see Figure~\ref{tbstar}). However, one can prove  that also in this case we can consider only applicative by-value contexts.
 We skip the details of such proof.

Moreover,
  we have $\K  \simeq_{\sigma I}^{*} \SSS (\K \K) (\SSS \K \K)$.  More in general, 
  the weak IPO bisimilarity  $ \simeq_{\sigma I}^{*}$ coincides with the lazy/cbv equivalence on CL:
  
\begin{thm}\label{ug}
For all $M,N \in CL^0$, \ \ $M\simeq_{\sigma I}^{*} N  \ \Longleftrightarrow\ M\simeq_{\sigma}N$.
\end{thm}
\begin{proof}
($\subseteq$) One can show that $\simeq_{\sigma I}^{*}{}\subseteq{} \simeq_{\sigma }^{*}$
by coinduction,  as in the proof of Lemma~\ref{ccont}, by showing that $\simeq_{\sigma I}^{*} $ is a  bisimulation on CL$^*$, also using Proposition~\ref{costar}. Then, since $ \simeq_{\sigma }^{*}$
coincides with $ \simeq_{\sigma }$ on CL-terms, we obtain the claim.
\\ 
($\supseteq$) By coinduction, showing that $\simeq_{\sigma}$ is a weak IPO bisimulation on CL$^*$.
\end{proof}

As a consequence of Theorem~\ref{eqv} and Theorem~\ref{ug} above, we can recover the lazy/cbv observational equivalence on $\lambda$-terms as  weak IPO bisimilarity on 
CL$^*$.

\begin{prop}
For all $M,N \in \Lambda^0$, \ \  
$M \approx_{\sigma} N \ \Longleftrightarrow\ \mathcal{ T}(M)\simeq^{*}_{\sigma I}   \mathcal{T}(N)$.
\end{prop}

However, such notions of weak IPO bisimilarities still suffer of the
problem of being infinitely branching, since the IPO contexts are $[\
]$, $[\ ]P$ for the lazy case, and $[\ ]$, $[\ ]V$ for the cbv case,
for all $P,V\in (CL^*)^0$.

This problem will be solved in the next section, where we introduce
the notion of second-order context category, and we endow CL$^*$ with
such a structure.

\section{Second-order Term Contexts}
\label{sec}

The definition of term context category \cite{LM00} can be generalized
to a definition of second-order term context category. The
generalization is obtained by extending the term syntax with function
(second-order) variables, that is, variables not standing for terms but
instead for functions on terms. The formal definition is the following

\begin{defi}[Category of Second-order Term Contexts]
Let $\Sigma$ be a signature for a term language. The category of
\emph{second-order term contexts} over $\Sigma$ is defined by: objects
are finite lists of naturals $ \langle n_1, \ldots , n_k \rangle$, an
arrow $ \langle m_1, \ldots, m_h \rangle \rightarrow \langle n_1,
\ldots , n_k \rangle$ is a k-tuple $\langle t_1 , \ldots , t_k
\rangle$, where the term $t_i$ is defined over the signature $\Sigma
\cup \{ F_1^{m_1}, \ldots , F_h^{m_h} \} \cup \{ X_{i,1} \ldots,
X_{i,n_i} \}$, where $F_i^{m_i}$ is a function variable of arity
$m_i$, $X_{i,j}$ is a ground variable.  The category of
\emph{second-order linear term contexts} is the subcategory whose
arrows are n-tuples of terms, satisfying the condition that the
n-tuples have to contain exactly one use of each function variable
$F^{m_i}_i$ and ground variable $X_{i,j}$.  The category of
\emph{second-order function-linear term contexts}, $T_2^* (\Sigma)$,
is the subcategory whose arrows are n-tuples of terms, satisfying the
condition that the n-tuples have to contain exactly one use of each
function variable $F^{m_i}_i$, moreover no function variable appears
inside the argument of another function variable.
\end{defi}

\noindent {\bf Remark.} Notice that the above definition of
second-order linear term contexts is different from that given in the
conference version of the present paper, \cite{DHL08}. The
modification was necessary because the original definition was
incorrect (second-order linear contexts were not closed by
composition).
\medskip

In the following we are going to use just a subcategory of the
category of second-order function-linear term contexts, however, at
this point, we prefer to present the original idea of second-order
term contexts in its full generality.

\begin{exa}
Given the signature of natural numbers $\{0, S, +\}$, examples of second-order linear contexts representing arrows in $\langle 2, 0 \rangle \rightarrow \langle 0, 2\rangle$ are:
\[
\!\!\langle F_2^0(), F_1^2(S(X_{2,2}) + X_{2,1}) \rangle, \   
\langle F_1^2(0,0), F_2^0() + (X_{2,1} + X_{2,2}) \rangle, \ 
\langle F_1^2(0, F_2^0()) , (X_{2,1} + X_{2,2})  \rangle
\]
Note that the last context is not function-linear. Examples of second-order function-linear contexts are:
\[ 
\langle F_2^0(), F_1^2(X_{2,2}, 0) \rangle, \  
\langle F_1^2(0,0), F_2^0()  +  X_{2,2} + X_{2,2} \rangle, \  
\langle F_1^2(0,0) + F_2^0(), X_{2,2} + X_{2,2}  \rangle
\]
None of the above contexts is linear.  Examples of second-order contexts that are neither function-linear nor linear are:
\[
\langle 0, F_1^2(X_{1,2}, X_{2,2})\rangle, \  
\langle F_1^2(0, F_2^0()), X_{2,2}\rangle, \  
\langle F_1^2(0,0), (F_2^0() + X_{1,2}) + (F_2^0() + X_{2,2})  \rangle 
\]
\end{exa}

Intuitively, an arrow in $\langle 2, 0 \rangle \rightarrow \langle 0, 2\rangle$ represents a pair of contexts containing  two holes $F_1^2, F_2^0$, where $F_1^2$ is a hole that must be filled by a term representing a function with two arguments while $F_2^0$ is a hole that must be filled by a term representing  function with no arguments, i.e., a ground term. The first context in the pair $\langle 2, 0 \rangle \rightarrow \langle 0, 2\rangle$ represents a function with no arguments, while the second context represent a function with two arguments $X_{2,1}, X_{2,2}$.

One can check that the standard category of term contexts over $\Sigma$ coincides with the subcategory whose objects are the lists containing only copies of the natural number $0$; in fact this subcategory uses function variables with no arguments and the ground variables do not appear. 

The identity arrow on the object  $\langle n_1, \ldots , n_k \rangle$ is: 
\[ \langle F_1^{n_1} (X_{1,1},\ldots X_{1, n_1} ), \ldots , F_k^{n_k} (X_{k,1},\ldots X_{k, n_k}) \rangle
\]

In order to define composition in the categories of second-order term contexts, it is convenient to consider the $\lambda$-closure of the tuple of terms representing arrows and to define arrow composition through  $\beta$-reduction.

The $\lambda$-closure of  a term $t$ on the signature $\Sigma \cup \{ F_1^{m_1}, \ldots , F_h^{m_h} \} \cup \{ X_{1}, \ldots, X_{n} \}$ is $\la  F_1^{m_1} \ldots F_h^{m_h}.  \la X_{1} \ldots X_{n} . t$, which, for brevity, can also be written  as $ \la \vec{F}. \la \vec{X}.t$.  
In general, given a second-order context $\langle t_1, \ldots, t_k\rangle: \langle m_1, \ldots, m_h \rangle\rightarrow \langle n_1, \ldots , n_k \rangle$, we consider the $\lambda$-closure: $\la \vec{F}. \langle \la \vec{X}_1.t_1, \ldots ,  \la \vec{X}_{k}.t_k \rangle$. 
The composition between the morphisms: 
\[ \la \vec{F}. \langle \la \vec{X}_{1} .s_1 , \ldots , \la \vec{X}_{k}.s_k \rangle : 
\langle l_1, \ldots, l_h \rangle \rightarrow \langle m_1, \ldots , m_k \rangle
\]
\[
\la \vec{G}. \langle \la \vec{Y}_{1}. t_1 , \ldots , \la \vec{Y}_{j}.t_j \rangle : 
\langle m_1, \ldots, m_k \rangle \rightarrow \langle n_1, \ldots , n_j \rangle
\]
is the $\beta$-normal form of the $\lambda$-expression
\[
\la \vec{F} . 
( \la \vec{G}. \langle  \la \vec{Y}_{1} .t_1, \ldots ,  \la \vec{Y}_{j}. t_j \rangle) 
( \la \vec{X}_{1} .s_1, \ldots , \la \vec{X}_{k}.s_k)  : 
\langle l_1, \ldots, l_h \rangle \rightarrow \langle n_1, \ldots , n_j \rangle
\]
To give an example, the composition between 
\[
\la F . \la X_1 . F(X_1, 0) : \langle 2 \rangle \rightarrow \langle 1 \rangle 
\ \ \ \mbox{ and } \ \ \
\la G . \la Y_1 Y_2 . G(S(Y_1)) + Y_2 : \langle 1 \rangle \rightarrow \langle 2 \rangle
\] 
is given by:
\[\eqalign{
 &\la F . (\la G . \la Y_1 Y_2 . G(S(Y_1)) + Y_2)  (\la X_1 . F(X_1, 0))\cr
 \rightarrow_{\beta} & \la F . \la Y_1 Y_2 . (\la X_1 . F(X_1, 0))(S(Y_1))) + Y_2)  \cr
 \rightarrow_{\beta} & \la F . \la Y_1 Y_2 . F( (S(Y_1), 0) + Y_2)
 : \langle 2 \rangle \rightarrow \langle 2 \rangle\ . 
  }
\]

In other words, the composition is given by a $j$-tuple of expressions $t_i$ in which every function variable $G_l$ is substituted by the corresponding expression $s_l$, with the ground variables of $s_l$ substituted by the corresponding parameters of $G_l$ in $t_i$. 

Note that the identity morphism is defined as a $\lambda$-term implementing the identity function,
while composition on morphisms is defined by the function composition in the $\lambda$-setting. Given this correspondence, it is easy to prove that the categorical properties for the identity hold, while  the associativity of composition essentially follows from the unicity of the normal form. 

Finally one need to prove that composition preserve linearity  and function-linearity.  For what concerns linearity, it is a well-known result that linear $\lambda$-terms are closed by $\beta$-reduction. From this fact one can immediately prove that second-order linear contexts are closed by composition. 

Preservation of function-linearity can be proved similarly. 
First we generalize the notion of function-linearity to $\lambda$-terms stating that  a function-linear $\lambda$-term is a typed
 lambda-term with constants, where 
\begin{enumerate}[$\bullet$]
\item 
all the variables and constants have either a ground type or a first-order function type; 
\item
each bound function variable (e.g. $F$) appears exactly once in the term, and only inside the arguments of constants (e.g. $S(F(0)+ 0$), or inside the arguments of $\lambda$-expressions having a second-order function type (e.g. $(\la G \la Y. G(Y) +Y) (\la X. F(X + S(0))) $). That is, no function variable appears inside the argument of an expression that has first order function type and is not a constant (e.g. $G(S(F(0))+ 0)$ and $(\la X. X +X)(F(0)) $).   
\end{enumerate}
It is straightforward to prove that function-linear $\lambda$-terms
are closed by $\beta$-reduction and that, given two function-linear
second-order contexts, the term, whose $\beta$-normal form defines
composition, is a function-linear $\lambda$-term.  From this the claim
follows.

The main general result on second-order term contexts is the following:
\begin{prop}\label{genrpo}
For any signature $\Sigma$, in the category of second-order (linear) (function-linear) term contexts over $\Sigma$, any commuting square, having as initial vertex the empty list $\epsilon$, has an RPO.
\end{prop}
\begin{proof}
First we present the proof for the special case useful in this paper, namely we consider the restricted category containing as objects the lists with at most one element. 
Given two arrows with domain the empty list: 
\newcommand{\vuoto}[1]{}
$\vuoto{\la \vec{X}_1 .} t_1 : \epsilon \rightarrow \langle n_1 \rangle$ and 
$\vuoto{\la \vec{X}_2 .} t_2 : \epsilon \rightarrow \langle n_2 \rangle$, and two arrows 
 $\vuoto{\la F_1^{n_1} . \vec{Y} .} s_1 : \langle n_1 \rangle \rightarrow \langle m \rangle$, 
 $\vuoto{\la F_2^{n_2} . \vec{Y} .} s_2 : \langle n_2 \rangle
\rightarrow \langle m \rangle$ completing $t_1$ and $t_2$ into a commuting square ($\vuoto{\la F_1^{n_1} . \vec{Y}_1 .} s_1 \circ \vuoto{\la \vec{X}_1 .} t_1  = \vuoto{\la F_1^{n_1} . \vec{Y}_1 .} s_1 \circ \vuoto{\la \vec{X}_1 .} t_1   : \epsilon \rightarrow \langle m \rangle$),
the corresponding RPO for this commuting square is inductively defined on the structures of $s_1$, $s_2$. 
There are several cases to consider:
 
\begin{enumerate}[(i)]
\item $s_1 = c_1(s_{1,1}, \ldots , s_{1,k_1})$ and $s_2 = c_2(s_{2,1},
\ldots , s_{2,k_2})$, with $c_1, c_2$ function symbols in the
signature $\Sigma$. Necessarily $c_1 = c_2$ (and $k_1 = k_2$). We have
to consider in which subterms of $s_1$ and $s_2$ the function
variables, $F_1^{n_1}$ and $F_2^{n_2}$, appear.  If $F_1^{n_1}$ and
$F_2^{n_2}$ appear in corresponding subterms, that is, there is an $i$
such that all $F_1^{n_1}$ appears in $s_{1,i}$ and all $F_2^{n_2}$ in
$s_{2,i}$, then we have that $\vuoto{\la F_1^{n_1} . \vec{Y} .}
s_{1,i}$ and $\vuoto{\la F_2^{n_1} . \vec{Y} .} s_{2,i}$, together
with $t_1, t_2$, form a commuting square, and the RPO, inductively
defined, for this second commuting square, immediately induces the RPO
for $s_1$ and $s_2$.  The subcase where $F_1^{n_1}$ and $F_2^{n_2}$ do
not appear in corresponding subterms is treated at point (iii).

\item $s_1 = F_1^{n_1}(s_{1,1}, \ldots , s_{1,n_1})$ and $s_2 =
F_2^{n_2}(s_{2,1}, \ldots , s_{2,n_2})$, and, for the general case,
$F_1^{n_1}$, $F_2^{n_2}$ not appearing in the subterms $s_{h,i}$.  In
this case, we have that 
\[t_1[s_{1_1}/X_{1,1}, \ldots,
s_{1,n_1}/X_{1,n_1}] = t_2[s_{2_1}/X_{2,1}, \ldots,
s_{2,n_2}/X_{2,n_2}]\ ,
\]
  that is, there is a unifier i.e., a substitution making $t_1$ and
  $t_2$ equal. Consider the most general unifier (mgu) for $t_1$ and
  $t_2$, this is given by tuples of terms, $s'_{1,1}, \ldots ,
  s'_{1,n_1}$ and $s'_{2,1}, \ldots , s'_{2,n_2}$, such that
  $t_1[s'_{1_1}/X_{1,1}, \ldots, s'_{1,n_1}/X_{1,n_1}] =
  t_2[s'_{2_1}/X_{2,1}, \ldots, s'_{2,n_2}/X_{2,n_2}]$.

  $\vuoto{\la F_1^{n_1} . \vec{Y'} .}  F_1^{n_1}(s'_{1,1}, \ldots ,
  s'_{1,n_1})\!:\!\langle n_1 \rangle \rightarrow \langle m' \rangle$
  and $\vuoto{\la F_2^{n_2} . \vec{Y'} .} F_2^{n_2}(s'_{2,1}, \ldots ,
  s'_{2,n_2})\!:\!\langle n_2 \rangle \rightarrow \langle m' \rangle$
  complete $t_1$ and $t_2$ into a commuting square that is also an
  RPO, in fact any other pair of arrows completing $t_1$ and $t_2$
  into a commuting square and factorizing the original one needs to be
  of the form $\vuoto{\la F_1^{n_1} . \vec{Y''} .}F_1^{n_1}(s''_{1,1},
  \ldots , s''_{1,n_1}): \langle n_1 \rangle \rightarrow \langle m''
  \rangle$ and $\vuoto{\la F_2^{n_2} . \vec{Y''} .}
  F_2^{n_2}(s''_{2,1}, \ldots , s''_{2,n_2}): \langle n_1 \rangle
  \rightarrow \langle m'' \rangle$, with the two sequences $\langle
  s''_{1,i} \rangle$ and $\langle s''_{2,i} \rangle$ defining a
  unifier for $t_1, t_2$.  The unique arrow factorizing the two
  commuting squares is $\vuoto{\la F^{m'} . \vec{Y'} .}
  F^{m'}(s'''_{1}, \ldots s'''_{m'})$, where $\langle s'''_i \rangle$
  is given by the mgu property.
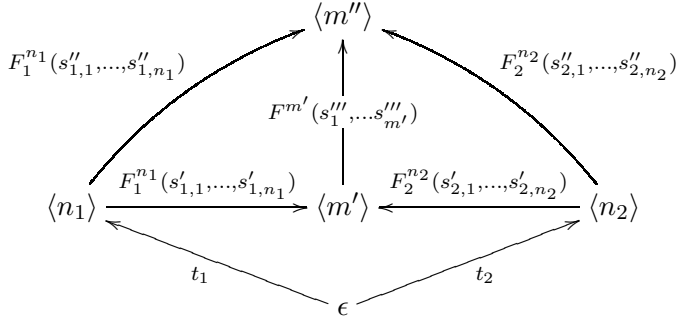
\[
\xymatrix@C+10 pt{ 
& &  \langle m'' \rangle & &
\\
& & & &
\\
 \langle n_1 \rangle 
  \ar[rr]^{\vuoto{\la F_1^{n_1} . \vec{Y'} .}F_1^{n_1}(s'_{1,1}, \ldots , s'_{1,n_1})}
  \ar@/^1em/[rruu]^{\vuoto{\la F_1^{n_1} . \vec{Y''} .}F_1^{n_1}(s''_{1,1}, \ldots , s''_{1,n_1})\quad}  
& & \langle m' \rangle 
   \ar[uu]|{\vuoto{\la F^{m'} . \vec{Y'} .} F^{m'}(s'''_{1}, \ldots s'''_{m'})} 
& & \langle n_2 \rangle 
    \ar[ll]_{\vuoto{\la F_2^{n_2} . \vec{Y'} .}F_2^{n_2}(s'_{2,1}, \ldots , s'_{2,n_2})} 
    \ar@/_1em/[lluu]_{\vuoto{\la F_2^{n_2} . \vec{Y''}  .}\quad F_2^{n_2}(s''_{2,1}, \ldots , s''_{2,n_2})}
  \\
& & \epsilon 
    \ar[llu]^{\vuoto{\la \vec{X}_1 .} t_1}   
    \ar[rru]_{\vuoto{\la \vec{X}_2 .} t_2 } 
& &
}
\]

\item In this point we consider all the remaining cases, that is,
where: $s_1 = c_1(s_{1,1}, \ldots , s_{1,k_1})$, $s_2 = c_2(s_{2,1},
\ldots , s_{2,k_2})$ and either $F_1^{n_1}$ and $F_2^{n_2}$ do not
appear in corresponding subterms, or $c_1 = F_1^{n_1}$ or $c_2 =
F_2^{n_2}$.  Let us consider the term $s'_1$ obtained from $s_1$ by
substituting any maximal subterm $s_o$ not containing $F_1^{n_1}$ by a
ground variable $X_{s_o}$. 

  For example, if $s_1 = c_1(s_{1,1},
c_2(s_{1,2,1}, F_1^{n_1}(s_{1,2,2,1}, s_{1,2,2,2}), s_{1,2,3}))$ then
$s'_1$ is the term $c_1(X_{s_{1,1}}, c_2(X_{s_{1,2,1}},
F_1^{n_1}(X_{s_{1,2,2,1}}, X_{s_{1,2,2,2}}), X_{s_{1,2,3}}))$, and
analogously for the term $s_2$.  Let $s''_1 = s'_1 \circ \vuoto{\la
\vec{X}_1 .} t_1$, and $s''_1 = s'_2 \circ \vuoto{\la \vec{X}_2 .}
t_2$.  Now we have that: $s''_1
[s_{1,\vec{l_1}}/X_{s_{1,\vec{l_{1}}}}, \ldots ,
s_{1,\vec{l_{m_1}}}/X_{s_{1,\vec{l_{m_1}}}}] = s''_2
[s_{2,\vec{j_1}}/X_{s_{1,\vec{j_{1}}}}, \ldots ,
s_{1,\vec{j_{m_2}}}/X_{s_{1,\vec{j_{m_2}}}}] $ that is, there exists a
unifier for $s''_1$ and $s''_2$, we can consider the most general
unifier, given by a pair tuples of terms $s'_{1,\vec{l_1}}, \ldots ,
s_{1,\vec{l_{m_1}}}$ and $s_{2,\vec{j_1}}, \ldots ,
s_{1,\vec{j_{m_2}}}$.  By repeating the arguments used at point (ii),
we have that $\vuoto{\la F_1^{n_1} . \vec{Y'} .}
s'_1[s'_{1,\vec{l_1}}/X_{s_{1,\vec{l_{1}}}}, \ldots ,
s'_{1,\vec{l_{m_1}}}/X_{s_{1,\vec{l_{m_1}}}}]$ and $\vuoto{\la
F_1^{n_1} . \vec{Y'} .} s'_2[s'_{2,\vec{j_1}}/X_{s_{1,\vec{j_{1}}}},
\ldots , s'_{1,\vec{j_{m_2}}}/X_{s_{1,\vec{j_{m_2}}}}]$ form an RPO.
\end{enumerate}
The proof for the general case is now almost immediate. The RPO for
the square
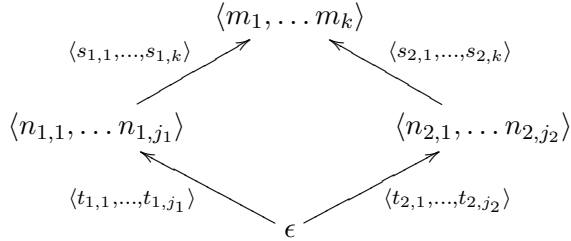
\[\xymatrix@R-15 pt@C-25 pt{ 
&&\langle m_{1}, \ldots m_{k} \rangle \cr\cr
  \langle n_{1,1}, \ldots n_{1,j_1} \rangle  
  \ar[rruu]^{\langle s_{1,1}, \ldots , s_{1,k} \rangle\quad}   &&&&  
  \langle n_{2,1}, \ldots n_{2,j_2} \rangle  
  \ar[lluu]_{\quad\langle s_{2,1}, \ldots , s_{2,k} \rangle}\cr\cr
&&\epsilon
   \ar[lluu]^{\langle t_{1,1}, \ldots , t_{1,j_{1}} \rangle\quad}   
   \ar[rruu]_{\quad\langle t_{2,1}, \ldots , t_{2,j_{2}} \rangle} 
\cr 
} 
\]
can be obtained by combining the RPO's for the $k$ diagrams
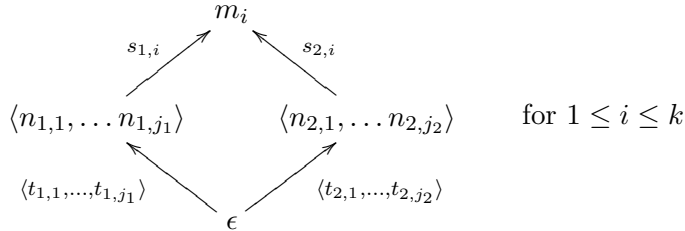
\[\vcenter{\xymatrix@R-15 pt@C-25 pt{ 
&&m_i \cr\cr
  \langle n_{1,1}, \ldots n_{1,j_1} \rangle  
  \ar[rruu]^{s_{1,i}}   &&&&  
  \langle n_{2,1}, \ldots n_{2,j_2} \rangle  
  \ar[lluu]_{s_{2,i}}\cr\cr
&&\epsilon
   \ar[lluu]^{\langle t_{1,1}, \ldots , t_{1,j_{1}} \rangle\quad}   
   \ar[rruu]_{\quad\langle t_{2,1}, \ldots , t_{2,j_{2}} \rangle} 
\cr 
}}\qquad\hbox{for $1 \leq i \leq k$}
\]
into a sequence. In turn, the RPO for these
diagrams can be obtained by essentially repeating the construction
presented for the unary case.  Finally, it is immediate to prove that
the presented construction preserve linearity and function-linearity
of arrows.
\end{proof}

The above proposition holds also for the case of linear second-order contexts and the prove remains almost the same.

\subsection{\texorpdfstring{CL$^*$}{CL*} as  Second-order Rewriting System}\label{clsec}
In this section, we consider the second-order context category for the combinatory
calculus  CL$^*$ and we show that the weak IPO lazy bisimilarity thus obtained coincides with the lazy
observational equivalence on $\lambda$-calculus, while for the cbv case we get a finer equivalence. 
Interestingly, the second-order open bisimilarity gives a uniform 
characterization also on open terms.
 
Note that the terms of CL are defined by the signature $\Sigma_{CL} = \{K, S, \app\}$, where $\app$ is the binary operation of application that is usually omitted.  So the term $\SSS \K \K$ actually stands for $\app ( \app (\SSS, \K), \K)$.

First we deal with the lazy case, then we will sketch also  the cbv case.

\subsubsection{The Lazy Second-order Reactive System}

\begin{defi}[Lazy  Second-order Reactive System on CL$^*$]
The \emph{lazy second-order reactive system} ${\mathbf C}_{l}^{2*}$ consists of:
\begin{enumerate}[$\bullet$]
\item the  function-linear category  whose objects are the lists with at most one element, 
and whose arrows $\epsilon \rightarrow \langle n \rangle$ are the terms of CL$^*$ with, at most,
$n$ (first order) metavariables,
\[ M^n \ ::=\   X_1\ | \ldots | X_n \  | \  \K\  |\  \SSS\  |\  \K' M^n \ |\   \SSS' M^n \ |\  \SSS'' M^n M^n \ |    \ M^n M^n \]
 and  whose arrows $\langle m \rangle  \rightarrow \langle n \rangle$ are
the second-order  contexts defined by:
\[ \CC^{m,n} \ ::= \ F(M^n_1, \ldots , M^n_m) \ | \ M^n\CC^{m,n} \ |\ \CC^{m,n} M^n \]
\item the reactive contexts are all the second-order applicative contexts of the shape
\[F(M^n_1, \ldots, M^n_m) N^n_1 \ldots N^n_k\ ;\]
\item the reaction rules  are
\[\K X_1  \rightarrow \K' X_1\quad\K' X_1 X_2 \rightarrow X_1\]
\[\SSS X_1  \rightarrow  \SSS' X_1\quad\SSS' X_1 X_2 \rightarrow  
  \SSS'' X_1 X_2\quad\SSS'' X_1 X_2 X_3 \rightarrow (X_1 X_3 )(X_2 X_3)
\]
where $\K X_1,  \SSS X_1 : \epsilon \rightarrow \langle 1 \rangle$, $\K' X_1 X_2,  \SSS' X_1 X_2 : \epsilon \rightarrow \langle 2 \rangle$ and $\SSS'' X_1 X_2 X_3 : \epsilon \rightarrow \langle 3 \rangle$.
\end{enumerate}
\end{defi}

\noindent
Second-order contexts as defined above can be represented by $C[F(M_1, \ldots , M_m)]$, where $C[\ ]$ is a
 unary  first-order context on CL$^*$ (with metavariables).
To maintain the notation for contexts used in Sections 4, 5,  in the sequel a second-order context 
$C[F(M_1, \ldots , M_m) ]:
\langle m\rangle \rightarrow \langle n \rangle$ will be more conveniently written
as  $  C[\ ]_\theta   $,
where $\theta$ is a substitution s.t.\ $\theta(X_i) = M_i$ for all $i=1, \ldots , m$, moreover we write $M \stackrel{ C[\ ]_\theta }{\rightarrow} M'$ \iff\ $C[M\theta]   \rightarrow M'$.  Given Proposition~\ref{genrpo}, and the underlined RPOs construction, we have:

\begin{cor}
The reactive system ${\mathbf C}_{l}^{2*}$ has redex RPOs.
\end{cor}


 
 
 \noindent \emph{Example}: Let $M=XM_1$. Some of the IPO reductions of $M$ 
 are the following:\\
$XM_1 \stackrel{  [\ ]_{\{\K/X\}}  }{\longrightarrow} \K' M_1$;\ \ \
$XM_1 \stackrel{  [\ ]_{ \{\K'Y/X\}} }{\longrightarrow} Y$; \ \ \
$XM_1 \stackrel{  [\ ] _{\{\K'/X \}}Y  }{\longrightarrow} M_1$; \ \ \
$XM_1 \stackrel{  [\ ] _{\{ \SSS/X \}}  }{\longrightarrow} \SSS' M_1$;\ \ \
$XM_1 \stackrel{  [\ ] _{\{ \SSS'Y/X\} }  }{\longrightarrow} \SSS'' YM_1$;\ \ \
$XM_1 \stackrel{  [\ ] _{\{ \SSS'/X\} } Y}{\longrightarrow} \SSS'' M_1Y$;\ \ \
$XM_1 \stackrel{  [\ ]_ {\{ \SSS''YZ/X\}}   }{\longrightarrow} (YM_1) (Z M_1)$; \\
$XM_1 \stackrel{  [\ ]_ {\{ \SSS''Y/X\}}  Z }{\longrightarrow} (YZ)(M_1 Z)$; \ \ \
$XM_1 \stackrel{  [\ ]_ {\{ \SSS''/X \} } YZ }{\longrightarrow} (M_1 Z)(Y Z)$; \ \ \
$XM_1 \stackrel{  [\ ]_{\{\K Y/X\}}  }{\longrightarrow} \K' Y M_1$;\ \ \ 
$XM_1 \stackrel{  [\ ]_{\{{\K} Y_1 Y_2/X\}}  }{\longrightarrow} \K' Y_1 Y_2 M_1$.
\\ Notice that $[\ ]_{ \{  {\K} Y_1 \ldots Y_n/X \} }$ is an IPO context for any $n$.
\vspace*{0.4cm}

In general, the IPO contexts are summarized in Figure~\ref{IPOdue}.

Using Proposition~\ref{list-weak-congruence}, we can prove that the weak IPO bisimilarity
$\simeq^{2*}_{l I}$ is a congruence, and it has a simpler characterization in terms of applicative
contexts. Namely, 
we can consider as list extension category the category of all function-linear term contexts.
In the alternative notation, a second-order linear term contexts can be written as $C[\__{\theta_1}, \ldots, \__{\theta_n}]$, where $C[\__1, \ldots , \__n]$ is a first-order multi-holed context and $\theta_1, \ldots , \theta_n$
are n substitutions, each one acting on the term put in the corresponding hole.
By repeating the arguments for the first-order case, one can show that any second-order linear term context either is IPO uniform or it has a
reactive index. Then, by Proposition~\ref{list-weak-congruence}, we have:

\smallskip

\begin{figure}
\begin{center}
\begin{tabular}{| l |l | l|}
\hline
term $M$ & IPO contexts  & reactive IPO contexts\\
\hline\hline
 $X$  &  $  [\ ]_{ \{ {\A} Y/X\}}$, $ [\ ]_{\{  {\A} /X\}} Y $,  
 ${\A} \vec{Y} C_1[\ ]_{\emptyset}$ &  $  [\ ]_{ \{ {\A} Y/X\}}$, $ [\ ]_{\{  {\A} /X\}} Y $
 \\
\hline
 $XP_0\vec{P}$ & $[\ ]_{ \{  {\A} \vec{Y} /X\}} $, 
 ${\A} \vec{Y} C_1 [\ ]_{\emptyset}$ & $[\ ]_{ \{  {\A} \vec{Y} /X\}} $
   \\
\hline
${\bf C} \vec{P}$ , $M$ value &  $[\ ]_{ \emptyset} X  $,   ${\A} \vec{Y} C_1 [\ ]_{\emptyset}$ &  $[\ ]_{ \emptyset} X  $
\\
\hline
${\bf C} \vec{P}$, $M$ not value & $[\ ]_{\emptyset}  $, ${\A} \vec{Y} C_1 [\ ]_{\emptyset}$ & $[\ ]_{\emptyset}  $ \\
\hline
\end{tabular} \bigskip
\\
where
\\
${\A} \in \{ \K, \SSS, \K'Z_1, \SSS'Z_1,  \SSS''Z_1 Z_2 \ |\ Z_1, Z_2 \mbox{ fresh} \}
$
\\
 ${\bf C} \in \{ \K, \SSS,  \K', \SSS', \SSS''\}$  
\\
$C_1[\ ]$ ranges over $C[\ ] \ ::= \ [\ ] \ |\ C[\ ]Z \ |\ ZC[\ ]$

\end{center}
\caption{Second-order IPO contexts for the lazy CL$^*$.}
\label{IPOdue}
\end{figure}

 
 \begin{prop}\label{duecon} 
\hfill \begin{enumerate}[\em(i)]
\item
   For all terms of $CL^*$ $M, N$, for any substitution $\theta$ and
  for any (possibly open) first-order context $C[ \ ]$, \[ M
  \simeq^{2*}_{l I} N \ \Longrightarrow \ C[M\theta] \simeq^{2*}_{l I}
  C[N\theta] \ .\]
\item
  $\simeq^{2*}_{l I}{}={}\simeq^{2*}_{l R}$, where $\simeq^{2*}_{l 
  R}$ denotes the weak IPO bisimilarity, where only reactive IPO
  contexts are considered (see Figure~\ref{IPOdue}).
\end{enumerate}
 \end{prop}

\noindent
By Proposition~\ref{duecon}(ii) above, the notion of IPO bisimilarity turns out to be much simpler, but it
 is  still infinitely branching  (when the term is of the shape $ X P_0 \vec{P}$ we have infinitely many IPO
 contexts $[\ ]_{\{ {\A} \vec{Y}/X\}}$). 
 However, one can prove that also the contexts $[\ ]_{\{ {\A} \vec{Y}/X\}}$, for any $|\vec{Y}|\geq 1$
 can be eliminated. This requires an ``ad-hoc'' reasoning:
 
 \begin{prop}\label{finlazy}
The lazy weak IPO bisimilarity $\simeq^{2*}_{l I}$  has a finitely branching characterization in terms of
the second-order IPO contexts of Figure~\ref{dappd}.
 \end{prop}
 \begin{proof}(sketch)
 Let $\simeq^{2*}_{l F}$ be the reduced bisimilarity obtained from $\simeq^{2*}_{l R}$ by not considering 
 the contexts $[\ ]_{\{ {\A} \vec{Y}/X\} }$, for any $|\vec{Y}|\geq 1$. Then $\simeq^{2*}_{l R}{}\subseteq{}
 \simeq^{2*}_{l F}$. In order to show the converse, one can first prove that the following is a weak IPO bisimulation:
 ${\mathrel R}
 = \{ (M', N') \mid \exists \theta.\ (M'\frown M\theta \ \wedge\ N' \frown N\theta \ \wedge\ M \simeq^{2*}_{l F} N
 \}$, where $M\frown N$ means that $M$ and $N$ are $KS$-convertible.
 \end{proof}
 
 \begin{figure}
\begin{center}
\begin{tabular}{| l |l |}
\hline
term $M$ & IPO contexts \\
\hline\hline
 $X$  &  $  [\ ]_{ \{ {\A}/X\}} Y$
  \\
\hline
 $XP_0\vec{P}$ & $[\ ]_{ \{  {\A}  /X\}} $
    \\
\hline
${\bf C} \vec{P}$ , $M$ value &  $[\ ]_{ \emptyset} X  $
\\
\hline
${\bf C} \vec{P}$ , $M$ not value & $[\ ]_{\emptyset}  $
 \\
\hline
\end{tabular} \bigskip
\\
where 
\\
${\A} \in \{ \K, \SSS, \K'Z_1, \SSS'Z_1,  \SSS''Z_1 Z_2 \ |\ Z_1, Z_2 \mbox{ fresh} \}$
\\
 ${\bf C} \in \{ \K, \SSS,  \K', \SSS', \SSS''\}$  

\end{center}
\caption{Finitely branching second-order IPO contexts for the lazy CL$^*$.}
\label{dappd}
\end{figure}
 
Finally, we are left to prove that the second-order weak IPO bisimilarity exactly recover the lazy 
observational equivalence. More in general, we will prove that the two equivalences coincide
on open terms.
Namely, we can view open terms with $n$ free variables as arrows from $\epsilon$ to $\langle n \rangle$
(by identifying variables with metavariables). Thus we have directly a notion of equivalence on open terms.
We will show that this equivalence coincides with the usual extension to open terms  of  the observational  equivalence by
 substitution. This gives a uniform finitely branching characterization of the 
observational equivalence on 
all (closed and open) terms.

\begin{prop}\label{inst}
For all $M,N \in \Lambda$, $M \widehat{\approx}_l N \ \Longleftrightarrow
 \ \mathcal{T}(M) \simeq^{2*}_{l I}    \mathcal{T}(N).$
\end{prop}

\subsection*{Proof of Proposition~\ref{inst}.}

We will show that  $\simeq^{2*}_{l I}$ coincides with
the natural extension to open terms of the first-order IPO bisimilarity $\simeq^{*}_{l I}$ of 
Section~\ref{first}.

\begin{defi}
Let $\widehat{\simeq}^{*}_{l I}$ be the extension of $\simeq^{*}_{l I}$ to open terms of CL$^*$
defined by, for all $M,N$ CL$^*$-terms such that $FV(M), FV(N) \subseteq \{ X_1 , \ldots , X_n\}$,
\[ M \widehat{\simeq}^{*}_{l I} N \  \mbox{iff}  \  \forall \theta :\{ X_1 , \ldots , X_n \} \rightarrow  
(CL^*)^0. \ M\theta
\simeq^{*}_{l I}  N\theta  \ .\]
\end{defi}

\begin{lem}
$\simeq^{2*}_{l R} {}\subseteq{} \widehat{\simeq}^{*}_{l R}$.
\end{lem}
\begin{proof}
We show that $\mathcal{R} = \{ (M\theta, N\theta) \mid M
\simeq^{\smash{2*}}_{l R} N \ \wedge\ M\theta, N\theta \in (CL^*)^0
\}$ is a first-order bisimulation.  From $M\simeq^{\smash{2*}}_{l R}
N$, by Proposition~\ref{duecon}, we have $M\theta\simeq^{2*}_{l R}
N\theta$. Assume
$M\theta\smash{\stackrel{\scriptscriptstyle[\ ]}{\rightarrow_{I}}}
M'$, since $M\theta\simeq^{2*}_{l I} N\theta$, then
$N\theta\smash{\stackrel{\scriptscriptstyle[\ ]}{\Rightarrow_I}} N' $,
$M' \simeq^{2*}_{R I} N'$ and $(M',N') \in \mathcal{R}$. Now assume
$M\theta\smash{\stackrel{\scriptscriptstyle[\ ]P}{\rightarrow_{I}}}
M'$, then
$M\theta\smash{\stackrel{\scriptscriptstyle[\ ]X}{\rightarrow_{I}}}
M''$ with $M''[P/X]=M'$, since $M\theta\simeq^{2*}_{l I} N\theta$ then
also
$N\theta\smash{\stackrel{\scriptscriptstyle[\ ]X}{\Rightarrow_{I}}}
N''$with $M'' \simeq^{2*}_{l I} N''$. Thus
$N\theta\smash{\stackrel{\scriptscriptstyle[\ ]P}{\Rightarrow_{I}}}
N'$ and $N'' [P/X] =N'$ is closed. Thus $(M',N') \in \mathcal{R}$.
\end{proof}

\begin{lem}\label{clred}
Let $M\in CL^*$, $M \rightarrow_l M'$. Then $M \widehat{\simeq}^{*}_{l I}  M'$.
\end{lem}
\begin{proof}
The proof follows from the fact that $\forall \theta. \ M\theta  \rightarrow^*_l M'\theta$ and 
$\simeq^{*}_{l I}$ is closed under $\rightarrow_l$.
\end{proof}

\begin{lem}
$\widehat{\simeq}^{*}_{l R} {}\subseteq{} \simeq^{2*}_{l R}$.
\end{lem}
\begin{proof}
We  show that $\mathcal{R} = \{ (M, N) \mid  M  \widehat{\simeq}^{*}_{l R}  N \}$ is a second-order bisimulation. 
If  $M \stackrel{[\ ]_\theta \vec{X}}{\rightarrow_I} M'$, then there are two cases.
\\
(i) $M= {\bf C} \vec{M}$, for a combinator ${\bf C}$ on $ CL^*$. Then $\theta = \emptyset$, and for 
any closing $\theta$ and closed $\vec{P}$ such that $|\vec{X}| = |\vec{P}|$, $M\theta \stackrel{\vec{P}}{\rightarrow_I} M''$ and $M'' = M' \theta [\vec{P}/X]$.
Since $M\theta \simeq^{*}_{l R} N\theta$, then $N\theta \stackrel{\vec{P}}{\Rightarrow_I} N''$ and 
$M'' \simeq^{*}_{l R}  N''$.  There are two subcases: either $\vec{X}= [\ ]$ or $\vec{X}= X$.
In the first subcase, we have $M \rightarrow_I M'$ (second-order) and $N \Rightarrow N$
 (second-order),
thus by Lemma~\ref{clred} $M' \widehat{\simeq}^{*}_{l R}  N$, and hence $(M',N)\in \mathcal{R}$.
In the second subcase, i.e., $\vec{X}=X$, $M$ is a value different from a variable, then one can check that also $N$ must
reduce to a value different from a variable, thus $N \stackrel{[\ ]_\emptyset  X}{\Rightarrow} N'$ and $N'' = N' \theta [P/X]$. Thus 
$M' \widehat{\simeq}^{*}_{l R} N' $, and hence $(M',N') \in \mathcal{R}$. 
\\
(ii) $M= X \vec{M}$.
Since for any closing $\theta$, $M\theta   \simeq^{*}_{l R}  N\theta$, then also 
$N \stackrel{[\ ]_\theta  \vec{X}}{\Rightarrow_I} N'$. Moreover, for any $\overline{\theta}$ closing 
$M\theta, N\theta$, for any $\vec{P}$ such that $|\vec{P}| =|\vec{X}|$, we have $M\theta \overline{\theta} \stackrel{\vec{P}}{\rightarrow_I} M''$,
 $N\theta \overline{\theta} \stackrel{\vec{P}}{\rightarrow_I} N''$, $M'' = M'\theta [\vec{P}/\vec{X}]$,
 $N'' = N'\theta [\vec{P}/\vec{X}]$. Thus for all $\overline{\theta}'.\ M' \overline{\theta}'  \simeq^{*}_{l R} 
 N \overline{\theta}'$, hence $(M',N') \in \mathrel{R}$.
\end{proof}

\subsubsection{The Cbv Second-order Reactive System}
The  main difference between the cbv and the lazy case is 
that the variables in the cbv case are meant to represent values, consequently  cbv substitutions 
have to map variables into values.

First of all, the values on CL$^*$ are defined by:
\[ V\ ::= \ X \ |\ \K \ |\  \K' V \ | \ \SSS\ |\  \SSS'V \ |\ \SSS'' VV \ .\]

\begin{defi}[Cbv  Second-order Reactive System on CL$^*$]
The \emph{cbv second-order reactive system} ${\mathbf C}_{v}^{2*}$ consists of:
\begin{enumerate}[$\bullet$]
\item the function-linear  category  whose objects are the lists with at most one element, 
and whose arrows $\epsilon \rightarrow \langle n \rangle$ are the terms of CL$^*$ with, at most,
$n$ (first order) metavariables, and  whose arrows $\langle m \rangle  \rightarrow \langle n \rangle$ are the second-order  contexts defined, briefly, by:
\[ \CC \ ::= \ F(V_1, \ldots , V_m) \ | \ M \CC \ |\ \CC M \]
where the values $V_1, \ldots ,V_m$ and the term $N$ are built using $n$ variables.
\item the reactive contexts are defined by
\[   \DD \ ::= \  F(V_1, \ldots , V_m) \ | \     \DD M \   | \     V \DD\ ;\]  
\item the reaction rules  are
\[\K X_1  \rightarrow \K' X_1\quad \K' X_1 X_2 \rightarrow X_1\]
\[\SSS X_1  \rightarrow  \SSS' X_1\quad
  \SSS' X_1 X_2 \rightarrow  \SSS'' X_1 X_2\quad
  \SSS'' X_1 X_2 X_3 \rightarrow (X_1 X_2 )(X_1 X_3)\ .
\]
\end{enumerate}
\end{defi}

\noindent
By Proposition~\ref{genrpo},  we have:

\begin{cor}
The reactive system ${\mathbf C}_{v}^{2*}$ has redex RPOs.
\end{cor}

As in the lazy case, a second-order context $\CC: \langle m\rangle
\rightarrow \langle n \rangle$ will be more conveniently denoted by $
C[\ ]_\theta $, where $C[\ ]$ is a unary first-order context and
$\theta$ is a cbv substitution, i.e., s.t.\ $\theta(X_i)$ is a value,
for all $i=1, \ldots , m$.

According to our definition, there are terms that are neither values
nor they are reducible (they do not contain any redex), the term $X Y$
is an example.  A term $M$ of this kind can be transformed in a
reducible one by substituting a single specific variable with a value.
We call \emph{critical variable} a variable of this kind.

\begin{defi}
The \emph{critical variable} of a second-order term $M$, $Cr(M)$, if it exists, is recursively defined by:
\[\eqalign{
 Cr(V)   &= \emptyset\ ,\cr
 Cr(X V) &= X\ ,\cr
 Cr(V M) &= Cr(M)\ , \mbox{if $M$ is not a value,}\cr
 Cr(M N) &= Cr(M)\ , \mbox{if $M$ is not a value.}\cr
  }
\]
\end{defi}
 
The second-order IPO contexts for cbv are summarized in Figure~\ref{cbvIPO}.
In that figure, the symbol $R$ ranges over most general reducible terms. That is, any reducible term can be obtained by instantiating the variables of a term contained in that grammar.
The symbol $T$ is used to represent general terms; remember that variables represent general values.  

\begin{figure}
\begin{center}
\begin{tabular}{| l |l | l|}
\hline
term $M$ & IPO contexts & reactive IPO contexts \\
\hline\hline
 $X$  
 &  $ [\ ]_{\{  {\A} /X\}} Y $, ${\A} [ \ ]_{\emptyset}$, $R C_1[\ ]_{\emptyset}$,
 &  $ [\ ]_{\{  {\A} /X\}} Y $, ${\A} [ \ ]_{\emptyset}$  \\   
\hline
$M$ a value but not a variable 
&  $[\ ]_{ \emptyset} X $, ${\A} [ \ ]_{\emptyset}$,   $RC_1 [\ ]_{\emptyset}$  
&  $[\ ]_{ \emptyset} X $, ${\A} [ \ ]_{\emptyset}$
\\
\hline
$M$ reducible
 & $[\ ]_{\emptyset}  $,  $RC_1 [\ ]_{\emptyset}$
 & $[\ ]_{\emptyset}  $
 \\
\hline
$M$ contains a critical variable
 & $[\ ]_{ \{  {\A} /Cr(M)\}} $,  $RC_1 [\ ]_{\emptyset}$ 
 & $[\ ]_{ \{  {\A} /Cr(M)\}} $\\
\hline
\end{tabular} \bigskip
\\
where
\\
${\A} \in \{ \K, \SSS, \K'X_1, \SSS'X_1,  \SSS''X_1 X_2 \ |\  X_1, X_2 \mbox{ fresh} \}$
\\
$R$ ranges over $R \ ::=\ \A Z \ |\ X R \ |\ R T$ 
\\
$C_1[\ ] $ ranges over $C[\ ] \ ::= \ [\ ] \ |\ C[\ ]T \ |\ TC[\ ]$
\\ with 
$T$ ranging over $T \ ::= \ X \ | \ (T T) $
\end{center}
\caption{Second-order IPO contexts for cbv CL$^*$.}
\label{cbvIPO}
\end{figure}
 
As for the previous case, by Proposition~\ref{list-weak-congruence} and by considering as list extension category the
category of all  by-value function-linear term contexts, we have:

\begin{prop} 
\hfill \begin{enumerate}[\em(i)] 
\item
  For all terms of $CL^*$ $M, N$, for any substitution $\theta$ and for
  any (possibly open) first-order context $C[ \ ]$, \[ M
  \simeq^{2*}_{v I} N \ \Longrightarrow \ C[M\theta] \simeq^{2*}_{v I}
  C[N\theta] \ .\]
\item 
  $\simeq^{2*}_{v I}{}={}\simeq^{2*}_{v R}$, where $\simeq^{2*}_{v R}$
   denotes the weak IPO bisimilarity, where only reactive IPO contexts
   are considered.
\end{enumerate}
 \end{prop}

\noindent
It is important to notice that the reactive IPO contexts provide directly  a finitely branching lts for the cbv combinatory logic (notice that, contrary to the lazy case, for the cbv case IPO contexts of the shape 
$[\ ]_{\{ {\A} \vec{Y}/X\}}$, for $|\vec{Y}|\geq 1$, do not exist, since  substitutions have to map 
variables into values).  

The cbv weak IPO bisimilarity turns out to be strictly included in the cbv contextual equivalence. Namely, if we consider
\[\eqalign{
  \mathcal{T} (\lambda x.x)
&=\SSS \K \K\ ,  \ \mbox{and}\cr
  \mathcal{T} (\lambda xy.xy)
&=\SSS [\SSS (\K\SSS)(\SSS (\K\K)(\SSS \K\K))] [\SSS (\SSS
(\K\SSS)(\K\K) )(\K\K)]  
  }
\]
then  $\mathcal{T} (\lambda x.x) \approx_v \mathcal{T} (\lambda xy.xy)$, however $\mathcal{T} (\lambda x.x) \not \simeq^{2*}_{vI}
\mathcal{T} (\lambda xy.xy)$, because
 \[ \mathcal{T} (\lambda xy.xy) \stackrel{[\ ]_{\emptyset}X}{\Rightarrow} \SSS'' (\K' X) (\SSS'' \K\K)
\stackrel{[\ ]_{\emptyset}Y}{\longrightarrow}, \mbox { while } \mathcal{T} (\lambda x.x) \stackrel{[\ ]_{\emptyset}X}{\Rightarrow} X 
\stackrel{[\ ]_{\emptyset}Y}{\not\Rightarrow} \ .\]
 The problem arises from the fact that in the second-order cbv bisimilarity we observe the existence of a 
critical variable, while in the contextual equivalence we do not.

\section{Final Remarks and Directions for Future Work}\label{final}
   There are several other attempts to deal with parametric rules in the literature. In his seminal paper \cite{Sew02}, Sewell presents two different constructions, one based on ground reaction rules and the other based on parametric rules.  The RPO construction can be seen as a categorical account of the ground rules construction.  Parametric rules, in the form they are defined in \cite{Sew02}, do not have an obvious categorical presentation. 
In \cite{KSS05}, the authors introduce the notion of \emph{luxes} to generalize the RPO approach to cases where the rewriting rules are given by pairs of arrows having a domain different from  $0$.  Luxes can be seen as a categorical account of the parametric rules approach of Sewell.  When instantiated to the category of contexts, the luxes approach allows to express rewriting rules not formed by pairs of ground terms but, instead formed by pairs of contexts (open terms), and so allowing parametricity. 
Compared to our approach, based on the notion of second-order context, the  approach of luxes is more abstract and it can be applied to a  wider range of cases (categories). However, if we compare the two approaches in the particular case of context categories, we find that the luxes approach has a more restricted way to instantiate a given parametric rule. This restriction results in a not completely satisfactory treatment of the $\lambda$-calculus.
It remains the open question of substituting the notion of second-order context with a more abstract and general one. This will allow to recover the extra generality of  luxes. 
 \\   A possible alternative approach for dealing with the $\lambda$-calculus in Leifer-Milner's RPO setting, it that of
    using 
suitable  encodings in the (bi)graph framework \cite{Mil06}. However, we feel that  our
 term  solution based on second-order context categories and  CL  is  
 simpler and more direct. Alternatively, in place of CL, one could also consider 
 a  $\lambda$-calculus with explicit substitutions, in order to obtain a convenient
  encoding of the $\beta$-rule, allowing for a representation as a second-order reactive system. 
  This is an experiment to be done. Here we have chosen CL, since it is simpler;  moreover,
  the correspondence between the standard $\lambda$-calculus and the one with explicit 
  substitutions deserves further study.
\\  We have considered lazy and cbv strategies, however also other strategies, e.g. 
\emph{head} and
\emph{normalizing} could be dealt with, possibly at the price of some complications due to the fact
that
 such strategies are usually defined on open terms. It would be
also interesting to explore non-deterministic strategies on $\lambda$-calculus.

\end{document}